\documentclass[runningheads]{llncs}
\usepackage[margin=1.38in]{geometry}
\usepackage[T1]{fontenc}
\usepackage{graphicx}
\usepackage[linesnumbered,ruled,vlined]{algorithm2e}
\SetCommentSty{footnotesize}
\SetKwInput{KwInput}{Input}
\SetKwInput{KwOutput}{Output}
\SetKwProg{Fn}{Function}{:}{\KwRet}
\usepackage{algpseudocode}

\usepackage{thmtools}
\usepackage{thm-restate}
\usepackage{xspace}
\usepackage{todonotes}
\usepackage{booktabs}
\usepackage{multirow}
\usepackage{mathtools}
\usepackage{amsfonts}

\usepackage{amsmath,amsthm}
\usepackage[hidelinks]{hyperref}
\usepackage{cleveref}
\usepackage{graphicx}
\usepackage[font=scriptsize,caption=false]{subfig}

\newcommand{\mpr}{minority population ratio\xspace}
\newcommand{\SW}{\mathrm{SW}}
\newcommand{\Util}{\mathrm{Util}}
\newcommand{\U}{\mathrm{U}}
\newcommand{\myparagraph}[1]{\medskip \noindent \textbf{#1}}
\newcommand{\fratio}{fair-ratio\xspace}
\newcommand{\Fratio}{Fair-ratio\xspace}
\newcommand{\AR}{\mathrm{FR}}
\newcommand{\DR}{\ensuremath{\mathrm{DRF}}}
\newcommand{\FOne}{\ensuremath{\mathrm{UNB}}}
\newcommand{\FTwo}{\ensuremath{\mathrm{BAL}}}
\newcommand{\FTwos}{\ensuremath{\mathrm{BAL}^*}}

\begin{document}
\title{Fair and Efficient Multi-Resource Allocation\\ for Cloud Computing}
%
%
\author{Xiaohui Bei\inst{1} \and
Zihao Li\inst{1} \and
Junjie Luo\inst{2}}

%
\authorrunning{X. Bei et al.}
%
\institute{
Nanyang Technological University, Singapore \and Beijing Jiaotong University, Beijing, China \\
\email{xhbei@ntu.edu.sg; zihao004@e.ntu.edu.sg; jjluo1@bjtu.edu.cn}}
\maketitle              
\begin{abstract}
	We study the problem of allocating multiple types of resources to agents with Leontief preferences.
	The classic Dominant Resource Fairness (DRF) mechanism satisfies several desired fairness and incentive properties, but is known to have poor performance in terms of social welfare approximation ratio.
	In this work, we propose a new approximation ratio measure, called \emph{\fratio}, which is defined as the worst-case ratio between the optimal social welfare (resp. utilization) among all \emph{fair} allocations and that by the mechanism, allowing us to break the lower bound barrier under the classic approximation ratio.
	We then generalize DRF and present several new mechanisms with two and multiple types of resources that satisfy the same set of properties as DRF but with better social welfare and utilization guarantees under the new benchmark.
	We also demonstrate the effectiveness of these mechanisms through experiments on both synthetic and real-world datasets.

\end{abstract}
%
%
%

\section{Introduction}
In order to offer flexible resources and economies of scale, in cloud computing systems,
a fundamental problem is to efficiently allocate heterogeneous computing resources, such as CPU time and memory, to agents with different demands.
This resource allocation problem presents several significant challenges from a technical perspective. For example, how to balance the efficiency of the system and fairness among users? How to incentivize agents to participate and truthfully reveal their private information? These are all delicate issues that need to be carefully considered when designing a resource allocation algorithm.


One of the most widely used mechanisms for multi-type resource allocation is the \emph{Dominant Resource Fairness (DRF)} mechanism proposed by \cite{Ghodsi2011}.
This work assumes that agents in the system have \emph{Leontief} preferences, which means they demand to receive resources of each type in fixed proportions.
Under such preferences, the proposed DRF mechanism generalizes the max-min allocation by equalizing the share of the most demanded resource, called \emph{dominant share}, for all agents.
\cite{Ghodsi2011} show that DRF satisfies a set of desirable properties. These include fairness properties: (\romannumeral1) \emph{share incentive (SI)}, all agents should be at least as happy as if each resource is equally allocated to all agents, and (\romannumeral2) \emph{envy-freeness (EF)}, no agent should prefer the allocation of another agent; efficiency properties: (\romannumeral3) \emph{Pareto optimality (PO)}, it is impossible to increase the allocation of one agent without decreasing the allocation of another agent; as well as incentive properties: (\romannumeral4) \emph{strategy-proofness (SP)}, no agent can benefit from reporting a false demand.
Consequently, DRF has received significant attention with many variants proposed to tackle different restrictions occurred in practice.


Despite the above attractive properties, however, DRF is known to have poor performance in terms of utilitarian social welfare, which is defined as the sum of utilities of all agents.
Many alternative mechanisms have then been proposed to tackle this issue and balance the trade-off between fairness and efficiency \cite{Joe-Wong2013,Bonald2014,Bonald2015,Jiang2021,Jin2016,Tang2016,Tang2020,Grandl2014}.
Most of these mechanisms still satisfy SI, EF, and PO. However, none of them satisfy SP.
Recently, \cite{Jiang2021} propose the so called 2-dominant resource fairness (2-DF) to balance fairness and efficiency.
Different from other mechanisms, 2-DF satisfies SP and PO, but does not satisfy SI and EF generally.
On the other hand, \cite{Parkes2015} justify this worst-case performance of DRF by showing that any mechanism satisfying any of the three properties SI, EF, and SP cannot guarantee more than $\frac{1}{m}$ of the optimal social welfare, which is also what DRF can achieve. Here $m$ denotes the number of resource types.
This characterization seems to suggest that from a worst-case viewpoint, DRF has the best possible social welfare guarantee among all fair or truthful mechanisms.

In this work, we aim to design new mechanisms that satisfy the same set of properties with DRF but with better efficiency guarantees.
In order to get around the theoretical barrier set by~\cite{Parkes2015}, we first propose and justify a new benchmark to measure the social welfare guarantee of a mechanism. Note that \cite{Parkes2015} and many other works use the \emph{approximation ratio}, which is defined as the worst-case ratio between the \emph{optimal social welfare among all allocations} and the mechanism's social welfare, as the performance measure of a mechanism.
However, since SI and EF are both fairness properties that place significant constraints on feasible allocations, it is not surprising that any allocation satisfying SI or EF would incur a large approximation ratio of $m$.
On the other hand, one can show that any mechanism satisfying SI has approximation ratio at most $m$.
This means all mechanisms satisfying SI and EF will have the same worst-case approximation ratio, which renders the approximation ratio notion meaningless in systems where these fairness conditions are hard constraints that must be satisfied.
Since fairness is a hard constraint in many practical applications,
we argue that it is more reasonable to compare the mechanism's social welfare to the optimal social welfare among \emph{all allocations that satisfy SI and EF}.
To this end, we modify the approximation ratio definition and propose this according variant. The new definition allows us to get pass the lower bound barrier from~\cite{Parkes2015} and design mechanisms with better social welfare approximation ratio guarantees.

\subsection{Our results}
We design new resource allocation mechanisms that satisfy properties such as SI, EF, PO, and SP, and at the same time achieve high efficiency.
The efficiency is measured by two objectives: \emph{social welfare}, defined as the sum of utilities of all agents, and \emph{utilization}, defined as the minimum utilization rate among all resources.
Social welfare is an indicator commonly used to measure efficiency, while improving utilization rate is also an important goal for cloud providers for cost-saving (see, e.g., Amazon\footnote{https://aws.amazon.com/blogs/aws/cloud-computing-server-utilization-the-environment/}, IBM\footnote{https://www.ibm.com/cloud/learn/cloud-computing}).
In academia, utilization has been studied by \cite{Li2017,Joe-Wong2013,Jin2018}.
For the performance measure, we define \emph{\fratio} for social welfare (resp. utilization) of a mechanism as the worst-case ratio between the social welfare (resp. utilization) achieved by the optimal mechanism \emph{satisfying SI and EF} and that by the mechanism.
See formal definitions in Section \ref{sec:prelim}.

We first focus on the setting where all agents' dominant resources fall into two types.
This is the most basic and arguably also the most important setting in cloud computing and other application domains such as high performance computing.
For example, most existing commercial cloud computing services, such as Azure, Amazon EC2, and Google Cloud, work with only two (dominant) resources: CPU and memory.
Two-resource setting can also be used to model the coupled CPU-GPU architectures where CPU and GPU are integrated into a single chip to achieve high performance computing \cite{Tang2016}.
In this setting, we present three new mechanisms $\FOne$, $\FTwo$, and $\FTwos$, all with better \fratio guarantees than DRF.
Different from DRF which equalizes the dominant share of all agents, the idea behind our new mechanisms is to partition all agents into two groups according to their dominant resources and carefully increase the share of agents with the smallest fraction of their non-dominant resource in each group.
Mechanism~$\FOne$ satisfies all four properties (SI, EF, PO, and SP) and has a \fratio of $\frac{3}{2}$ for social welfare and $2$ for utilization.
Mechanism $\FTwo$ further improves the \fratio for social welfare to $\frac{4}{3}$.
However, $\FTwo$ satisfies SI, EF, and PO, but not SP.
Finally, we generalize $\FTwo$ to a new mechanism $\FTwos$ which satisfies all the four properties and has the same asymptotic \fratio as $\FTwo$ when the number of agents $n$ goes to infinity.
We further provide a more fine-grained analysis of the \fratio parameterized by a \emph{\mpr} parameter $\alpha \in (0,\frac{1}{2}]$, which is defined as the fraction of agents in the smaller group classified by their dominant resources.
Table \ref{tab:result} lists a summary of the \fratio{s} of different mechanisms in the worst case and in terms of $\alpha$.
We also compare our mechanisms with DRF by conducting experiments on both synthetic and real-world data. Our results  match well with the theoretical bounds of \fratio{s} and show that both $\FOne$ and $\FTwos$ achieve better social welfare and utilization than DRF.

\begin{table}[t]
\caption{\Fratio results for $m=2$ resources overview.}
\label{tab:result}
\centering
\setlength{\tabcolsep}{20pt}
\renewcommand{\arraystretch}{1.7}
\begin{tabular}{@{}lll@{}}
\toprule
                       & Social Welfare                                          & Utilization                            \\ \midrule
 DRF (Lemma \ref{lem:n-2-DRF})     & $2$ \quad $(2-\alpha)$                                        & $\infty$ \quad $(\frac{1}{\alpha})$          \\
$\FOne$ (Theorem \ref{thm:n-2-F1})  & $\frac{3}{2}$ \quad $(1+\alpha)$                              & $2$ \quad $(\frac{1}{1-\alpha})$             \\
$\FTwo$ (Theorem \ref{thm:n-2-F2})  & $\frac{4}{3}$ \quad $(\frac{4-2\alpha}{3-\alpha})$           & $2$ \quad $(\frac{2}{1+\alpha})$ \\
$\FTwos$ (Theorem \ref{thm:n-2-F2s})& $\left[\frac{4-2\alpha}{3-\alpha}, \frac{4-2\alpha}{3-\alpha-\frac{1}{n}}\right]$ &
$\left[\frac{2}{1+\alpha}, \frac{2}{1+\alpha-\frac{1}{n}}\right]$ \\
[0.5em]
\bottomrule
\end{tabular}
\end{table}

\begin{table}[t]
\caption{Price of SP results overview.}
\label{tab:SP-result}
\centering
\setlength{\tabcolsep}{23pt}
\renewcommand{\arraystretch}{1.7}
\begin{tabular}{@{}lll@{}}
\toprule
& Social Welfare     & Utilization     \\
\midrule
$m=2$ (Theorem \ref{thm:n-2-SP})  &  $[1,3-\sqrt{3}+\frac{1}{2n}]$  & $[1,\frac{3}{2-\frac{1}{n}}]$  \\
$m=3$ (Theorem \ref{thm:lower-bound})  &  $[2,3]$   & $\infty$  \\
$m \ge 4$ (Theorem \ref{thm:lower-bound}) &   $m$     & $\infty$  \\
\bottomrule
\end{tabular}
\end{table}

Next we move to the general situation with $m \ge 2$ resources.
We first give a family $\mathcal{F}$ of mechanisms, containing DRF as a special case, that satisfy all the four properties.
This answers the question posed by \cite{Ghodsi2011} that \emph{``whether DRF is the only possible strategy-proof policy for multi-resource fairness, given other desirable properties such as Pareto efficiency''}.
Unfortunately, as we will see in the next part, for general $m$ all mechanisms that satisfy the four properties will have the same \fratio as DRF.
Nevertheless, we show that a generalization of $\FOne$ still satisfies the four properties and its \fratio is always weakly better than DRF.


Finally, we investigate the efficiency loss caused by incentive constraints.
We define the \emph{price of strategyproofness} (Price of SP) for social welfare (resp. utilization) as the best \fratio for social welfare (resp. utilization) among all mechanisms which satisfy SI, EF, PO, and SP.
Our results are summarized in Table 2.
For the case with $m=2$ resources, we show that the price of SP is at most $3 - \sqrt{3}+\frac{1}{2n}$ for social welfare, and at most $\frac{3}{2-\frac{1}{n}}$ for utilization.
When $m = 3$, the price of SP is between $2$ and $3$ for social welfare and $\infty$ for utilization.
Finally, when $m \ge 4$, the price of SP is $m$ for social welfare and $\infty$ for utilization, which implies that in the general setting all mechanisms that satisfy the four properties have the same \fratio as DRF.

\subsection{Related work}
Since its introduction by \cite{Ghodsi2011}, DRF has been extended in multiple directions, including the setting with weighted agents or indivisible tasks
\cite{Parkes2015}, the setting when resources are distributed over multiple servers with placement constraints \cite{Tahir2015,Wang2016} or without placement constraints \cite{Friedman2014,Wang2014}, a dynamic setting when agents arrive at different times \cite{Kash2014} and the case when agents' demands are limited \cite{Li2017,Narayana2021}.
In contrast to these works, we consider the original setting and aim to design mechanisms with better efficiency guarantees than DRF.
Notably, \cite{Li2017} generalize DRF to the limited demand setting, and study the approximation ratio of the generalized mechanism by comparing it with the optimal allocation satisfying PO, SI and EF.
Essentially, their results implies that for two resources, the \fratio of DRF is 2 for social welfare and $\infty$ for utilization, which can be seen as a special case of our more fine-grained result in Lemma \ref{lem:n-2-DRF} parameterized by $\alpha$.
\cite{Dolev2012} advocate a different fairness notion called \emph{Bottleneck Based Fairness (BBF)} for multi-resource allocation with Leontief preferences and show that a BBF allocation always exists.
\cite{Gutman2012} extend DRF and BBF for a larger family of utilities and give a polynomial time algorithm to compute a BBF solution.
Characterization of mechanisms satisfying a set of desirable properties under Leontief preferences has been studied in economics literature \cite{Nicolo2004,Friedman2011,Li2013}.
However, they consider different properties than what we consider.

\section{Preliminaries}\label{sec:prelim}


\subsection{Multi-resource allocation}
We start by introducing the formal model of multi-resource allocation. The notations are mainly adopted from \cite{Parkes2015}.
Given a set of agents $N=\{1,2,\dots,n\}$ and a set of resources $R$ with $|R|=m$, each agent $i$ has a \emph{resource demand vector} $\mathbf{D}_i=\{D_{i1},D_{i2},\dots,D_{im}\}$, where $D_{ir}$ is the ratio between the demand of agent $i$ for resource $r$ to complete one task and the total amount of that resource.
The \emph{dominant resource} of an agent $i$ is the resource $r_i^*$ such that $r_i^* \in \arg \max_{r \in R}D_{ir}$.
For simplicity, we assume that all agents have positive demands, i.e., $D_{ir}>0, \forall i \in N, \forall r \in R$.
For each agent $i$ and each resource $r$, define $d_{ir}=\frac{D_{ir}}{D_{ir_i^*}} \in (0,1]$ as the \emph{normalized demand} and denote the normalized demand vector of agent $i$ by $\mathbf{d}_i=\{d_{i1},d_{i2},\dots,d_{im}\}$.
An \emph{instance} of the multi-resource allocation problem with $n$ agents and $m$ resource is a matrix $\mathbf{I}$ of size $n \times m$ with each row representing a normalized demand vector.

To help better understand these notions, consider a cloud computing scenario where two agents share a system with 9 CPUs and 18 GB RAM.
Each task agent 1 runs require $\langle \textrm{1 CPUs}, \textrm{4 GB}\rangle$, and each task agent 2 runs require $\langle \textrm{3 CPUs}, \textrm{1 GB}\rangle$.
Since each task of agent 1 demands $\frac{1}{9}$ of the total CPU and $\frac{2}{9}$ of the total RAM, the demand vector for agent 1 is $\mathbf{D}_1=\{\frac{1}{9},\frac{2}{9}\}$, with RAM being its dominant resource, and the corresponding normalized demand vector is $\mathbf{d}_1=\{\frac{1}{2},1\}$.
Similarly, for agent 2 we have $\mathbf{D}_2=\{\frac{1}{3},\frac{1}{18}\}$, $\mathbf{d}_2=\{1,\frac{1}{6}\}$, and its dominant resource is CPU.

Given problem instance $\mathbf{I}$, an \emph{allocation} $\mathbf{A}$ is a matrix of size $n \times m$ which allocates a fraction $A_{ir}$ of resource $r$ to agent $i$.
We assume all resources are divisible.
An allocation $\mathbf{A}$ is \emph{feasible} if no resource is required more than available, i.e., $\sum_{i \in N} A_{ir} \le 1, \forall r \in R$.
We assume agents have \emph{Leontief preferences} and the \emph{utility} of an agent with its allocation vector $\mathbf{A}_i$ is defined as
\[
u_i(\mathbf{A}_i)=\max \{y \in \mathbb{R_+} : \forall r \in R, A_{ir} \ge y \cdot d_{ir}\}.
\]
We say an allocation is \emph{non-wasteful} if for each agent $i \in N$ there exists $y \in \mathbb{R}_+$ such that $A_{ir}=y \cdot d_{ir}, \forall r \in R$.
In words, for each agent, the amount of allocated resources are proportional to its normalized demand vector.
The \emph{dominant share} of an agent $i$ under a non-wasteful allocation $\mathbf{A}$ is $A_{ir_i^*}$, where $r_i^*$ is $i$'s dominant resource.

Denote the set of all instances by $\mathcal{I}$, and the set of all feasible allocations by $\mathcal{A}$.
A \emph{mechanism} is a function $f: \mathcal{I} \rightarrow \mathcal{A}$ that maps every instance to a feasible allocation.
We use $f_i(\mathbf{I})$ to denote the allocation vector to agent $i$ under instance $\mathbf{I}$.
A mechanism is non-wasteful if the allocation of the mechanism on any instance is non-wasteful.
We only consider non-wasteful mechanisms.

\subsection{Dominant Resource Fairness (DRF)}
The \emph{DRF} mechanism \cite{Ghodsi2011} works by maximizing and equalizing the dominant shares of all agents, subject to the feasible constraint.
Let $x$ be the dominant share of each agent, DRF solves the following linear program:
\begin{alignat*}{2}
  & \text{maximize}   & \quad & x          \nonumber \\
  & \text{subject to} &       & \sum_{i \in N} x \cdot d_{ir} \le 1, \quad \forall r \in R
\end{alignat*}

This linear program can be rewritten as $x^*=\frac{1}{\max_{r\in R} \sum_{i\in N} d_{ir}}$.
Then, for agent $i$ the allocation $\mathbf{A}_i=x^* \cdot \mathbf{d}_i$.


\subsection{Properties of mechanisms}
In this work we are interested in the following properties of a resource allocation mechanism.

\begin{definition}[Share Incentive (SI)]
An allocation $\mathbf{A}$ is SI if $u_i(\mathbf{A}_i) \ge \frac{1}{n}, \forall i \in N$.
A mechanism $f$ is SI if for any instance $\mathbf{I} \in \mathcal{I}$ the allocation $f(\mathbf{I})$ is SI.
\end{definition}

\begin{definition}[Envy Freeness (EF)]
An allocation $\mathbf{A}$ is EF if $u_i(\mathbf{A}_i) \ge u_i(\mathbf{A}_j), \forall i,j \in N$.
A mechanism $f$ is EF if for any instance $\mathbf{I} \in \mathcal{I}$ the allocation $f(\mathbf{I})$ is EF.
\end{definition}

\begin{definition}[Pareto Optimality (PO)]
An allocation $\mathbf{A}$ is PO if it is not dominated by another allocation $\mathbf{A}'$, i.e., there is no $\mathbf{A}'$ such that $\exists i_0 \in N: u_{i_0}(\mathbf{A}'_{i_0}) > u_{i_0}(f_{i_0}(\mathbf{I}))$ and $\forall i \in N: u_i(\mathbf{A}'_{i_0}) \ge u_i(f_i(\mathbf{I}))$.
A mechanism $f$ is PO if for any instance $\mathbf{I} \in \mathcal{I}$ the allocation $f(\mathbf{I})$ is PO.
\end{definition}

\begin{definition}[Strategyproofness (SP)]
A mechanism $f$ is SP if no agent can benefit by reporting a false demand vector, i.e., $\forall \mathbf{I} \in \mathcal{I}, \forall i \in N, \forall \mathbf{d}'_i, u_i(f_i(\mathbf{I})) \ge u_i(f_i(\mathbf{I}'))$, where $\mathbf{I}'$ is the resulting instance by replacing agent $i$'s demand vector by $\mathbf{d}'_i$.
\end{definition}

Notice that SI, EF, and PO are defined for both allocations and mechanisms, while SP is only defined for mechanisms.
It is easy to verify that a non-wasteful mechanism satisfies PO if and only if at least one resource is used up in the allocation returned by the mechanism.

\cite{Ghodsi2011} shows that DRF satisfies all of these desirable properties.

\subsection{Approximation ratio}

We define \emph{social welfare (SW)} of an allocation $\mathbf{A}$ as the sum of the utilities of all agents,
\[
\SW(\mathbf{A})=\sum_{i \in N} u_i(\mathbf{A}_i).
\]
As in \cite{Li2017}, we define \emph{utilization} of an allocation $\mathbf{A}$ as the minimum utilization rate of $m$ resources,
\[
\U(\mathbf{A})=\min_{r \in R} \sum_{i \in N} A_{ir}.
\]


As discussed in the introduction, we use a revised notion of approximation ratio to measure the efficiency performance of a mechanism, where we use the optimal \emph{fair} allocation as the benchmark instead of the original benchmark which is based on the optimal allocation.
\begin{definition}
The \emph{\fratio for social welfare (resp. utilization)} of a mechanism $f$ is defined as, among all instances $\mathbf{I} \in \mathcal{I}$, the maximum ratio of the optimal social welfare (resp. utilization) among all allocations that satisfy SI and EF
over the social welfare (resp. utilization) of $f(\mathbf{I})$, i.e.,
\[
\AR_{\SW}(f) = \max_{\mathbf{I} \in \mathcal{I}}\frac{\max\limits_{\mathbf{A} \textrm{ is SI,EF}}\SW(\mathbf{A})}{\SW(f(\mathbf{I}))}
\textrm{\quad and \quad}
\AR_{\Util}(f) = \max_{\mathbf{I} \in \mathcal{I}}\frac{\max\limits_{\mathbf{A} \textrm{ is SI,EF}}\U(\mathbf{A})}{\U(f(\mathbf{I}))}.
\]

\end{definition}

%

\section{Two Types of Resources}
In this section we focus on the case where there are only two competing resources. More specifically, we assume that among the $m$ types of resources, there exists $r_1, r_2 \in R$, such that for any agent $i$ and any other resource $r \neq r_1, r_2$, we have $d_{ir_1} \geq d_{ir}$ and $d_{ir_2} \geq d_{ir}$. This means in any allocation, other resources will not run out before $r_1$ or $r_2$ runs out. Thus it is equivalent to assume that $R$ contains only two resources $r_1$ and $r_2$.

%
%
%
%
%
%
%


We partition all agents into two groups $G_1$ and $G_2$, where $G_i (i=1,2)$ consists of all agents whose dominant resource is $r_i$.
Agents with demand vector $(1,1)$ are considered to be in $G_1$.
Denote $n_1=|G_1|$ and $n_2=|G_2|$.
Without loss of generality, we assume that $n_1 \ge \frac{n}{2}$ (otherwise we can rename the two resources).

We now let
\[
\alpha \coloneqq \frac{n_2}{n} \in (0,\frac{1}{2}]
\]
be the fraction of agents in the smaller group and we call $\alpha$ the \emph{\mpr}.
We assume that $\alpha>0$, because when $\alpha=0$ the only allocation satisfying SI is to give every agent $\frac{1}{n}$ of the first resource (and the corresponding amount of the second resource).
As we will see in the following, $\alpha$ is crucial in analyzing the \fratio of a mechanism.

We start by analyzing the \fratio of DRF.

\begin{restatable}{lemma}{DRF}\label{lem:n-2-DRF}
With $2$ resources, for instances with \mpr $\alpha$, we have
\[\AR_{\SW}(\DR) = 2-\alpha \textrm{\quad and \quad} \AR_{\Util}(\DR) = \frac{1}{\alpha}.\]
\end{restatable}

\begin{proof}
We first show $\AR_{\SW}(\DR) = 2-\alpha$.
For the lower bound, we build an instance with \mpr $\alpha$ and $n$ agents as follows.
The first group $G_1$ consists of $n(1-\alpha)$ agents who have the same demand vector $(1,\varepsilon)$, where $\varepsilon=\frac{1}{n}$.
The second group $G_2$ consists of $n\alpha$ agents, where except for one special agent $i^*$ whose demand vector is $(\frac{\varepsilon}{2},1)$, all other agents have the same demand vector $(1-\varepsilon,1)$.
We choose $n$ large enough such that $n\alpha \ge 2$.
The idea is that under DRF since all agents must have the same dominant share, their dominant share is close to $\frac{1}{n}$ because of the limit of resource~1 and hence SW will be close to 1, while there exists an allocation that satisfies SI and EF, and has SW close to $2-\alpha$ by giving roughly $1-\alpha$ dominant share to the special agent $i^*$
and $\frac{1}{n}$ dominant share to other agents.

Formally, under DRF, resource~1 will be used up and the dominant share of every agent is
\[
\frac{1}{n(1-\alpha)+(n\alpha-1)(1-\varepsilon)+\frac{\varepsilon}{2}} \le \frac{1}{n-2},
\]
so SW of the DRF allocation is at most $\frac{n}{n-2}$.
However, if we give $\frac{1}{n}$ dominant share to every agent except for agent $i^*$ and give $i^*$ the bundle $(x\frac{\varepsilon}{2},x)$, where $x=(1-\alpha)(1-\frac{1}{n})+\frac{1}{n}$, such that resource~2 is used up, then the SW is
\[
1-\frac{1}{n}+x \ge 2-\alpha-\frac{1}{n}.
\]
It is easy to verify that the above allocation, denoted by $\mathbf{A}^*$, satisfies SI, EF (and PO).
For EF, notice that the special agent $i^*$ receives $x\frac{\varepsilon}{2}\le\frac{1}{2n}$ of resource~1 while all other agents in $G_2$ receive $\frac{1-\varepsilon}{n}=\frac{1-\frac{1}{n}}{n}$ of resource~1.
Thus we have 
\[
\AR_{\SW}(\DR) \ge 
\frac{2-\alpha-\frac{2}{n}}{\frac{n}{n-2}}\overset{n \to \infty}{\longrightarrow} 2-\alpha.
\]

For the upper bound, for any instance $\mathbf{I}$ with SW $s \ge 1$ under DRF, we show that SW of any allocation satisfying SI is upper bounded by $(2-\alpha)s$, and hence $\AR_{\SW}(\DR) \le 2-\alpha$.
Let $\mathbf{A}$ be an arbitrary allocation on $\mathbf{I}$ satisfying SI.
If resource~1 is used up under DRF, then agents in $G_1$ get $(1-\alpha)s$ of resource~1 and agents in $G_2$ get $1-(1-\alpha)s$ of resource~1.
Notice that under DRF every agent gets $\frac{s}{n}$ dominant share while in $\mathbf{A}$ every agent gets at least $\frac{1}{n}$ dominant share.
Thus, in $\mathbf{A}$ agents in $G_2$ get at least $\frac{1}{s}(1-(1-\alpha)s)$ of resource~1.
Consequently, in $\mathbf{A}$ agents in $G_1$ get at most $1-\frac{1}{s}(1-(1-\alpha)s)$ of resource~1.
Then we have
\[
\SW(\mathbf{A}) \le 1-\frac{1}{s}(1-(1-\alpha)s)+1=3-\alpha-\frac{1}{s} \le (2-\alpha)s,
\]
where the last inequality follows by $s \ge 1$.
Analogously, if resource~2 is used up under DRF, we have that $\SW(\mathbf{A}) \le 3-(1-\alpha)-\frac{1}{s}\le (2-\alpha)s$.

We then show $\AR_{\Util}(\DR) = \frac{1}{\alpha}$.
For the lower bound, we use the same instance used for the lower bound of SW.
Recall that for that instance under DRF resource~2 is not used up, and the dominant share of every agent is at most $\frac{1}{n-2}$, so at most $\frac{n\alpha+(1-\alpha)n\varepsilon}{n-2} \le \alpha+\frac{2}{n-2}$ of resource~2 is used.
However, in $\mathbf{A}^*$ resource~1 is not used up and at least $1-\alpha+(n\alpha-1)\frac{1-\varepsilon}{n} \ge 1-\frac{2}{n}$ of resource~1 is used.
So
\[
\AR_{\Util}(\DR) \ge 
\frac{1-\frac{2}{n}}{\alpha+\frac{2}{n-2}}\overset{n \to \infty}{\longrightarrow} \frac{1}{\alpha}.
\]
For the upper bound, since DRF satisfies SI, at least $1-\alpha$ of resource~1 is used and at least $\alpha$ of resource~2 is used, so $\AR_{\Util}(\DR) \le \max\{\frac{1}{1-\alpha},\frac{1}{\alpha}\}=\frac{1}{\alpha}$.
\end{proof}

When $\alpha$ approaches 0, we have $\AR_{\SW}(\DR) \rightarrow 2$ and $\AR_{\Util}(\DR) \rightarrow \infty$.
Notice that with 2 resources $\AR_{\SW}(f)$ for any mechanism $f$ satisfying SI is at most $2$ as the mechanism can always achieve at least $1$ in SW.

In the following, we present two new mechanisms with the same set of properties as DRF but with better \fratio{s}.

\subsection{Mechanism $\FOne$}

\begin{algorithm}[t]
\DontPrintSemicolon
$C\leftarrow (c_1, c_2) = (1,1)$  \tcp*{remaining resources}
$G_1\leftarrow\{i \mid d_{i,1}=1\}$; $G_2\leftarrow\{i \mid d_{i,1}<1\}$  \\
\ForEach{$i \in N$}
{
  $\mathbf{A}_i\leftarrow\frac{1}{n} \mathbf{d}_i$ \tcp*{every agent receives $\frac{1}{n}$ dominant share}
  $C \leftarrow C-\mathbf{A}_i$
}
\While{$c_1 > 0$ and $c_2 > 0$}
{
  $P \leftarrow \arg\min_{i \in G_2} A_{i,1}$  \tcp*{agents with the smallest fraction of resource $r_1$}
  $\delta_0 \leftarrow \min\limits_{i \in N \setminus P} A_{i,1}-\min\limits_{i \in P} A_{i,1}$  \tcp*{increasing step when 2nd smallest fraction of resource $r_1$ is reached}
  $\delta_1 \leftarrow \frac{C_1}{|P|}$,
  $\delta_2 \leftarrow \frac{C_2}{\sum_{i \in P} \frac{1}{d_{i,1}}}$
  \tcp*{increasing step when resource $r_1$ (or $r_2$) is used up}
  $\delta^* \leftarrow \min\{\delta_0, \delta_1, \delta_2\}$ \\
  \ForEach{$i \in P$}
    {
      $\mathbf{A}_i \leftarrow \mathbf{A}_i+(\delta^*,\frac{\delta^*}{d_{i,1}})$
      \tcp*{increase resource $r_1$ by the same $\delta^*$ for agents in $P$}
      $C \leftarrow C-(\delta^*,\frac{\delta^*}{d_{i,1}})$
    }
}
\Return $\mathbf{A}$
\caption{$\FOne(\mathbf{d}_1,\mathbf{d}_2,\dots,\mathbf{d}_n)$}
\label{alg:F1}
\end{algorithm}

The 
lower bound analysis in the proof
of Lemma \ref{lem:n-2-DRF} shows that when the population of two groups are unbalanced, i.e., when $\alpha$ is close to $0$, it is better to allocate more resources to agents in the minor group $G_2$ with smaller $d_{i,1}$.
This idea leads to mechanism $\FOne$, described in Algorithm~\ref{alg:F1}.
The mechanism has two steps.
In step 1, the mechanism allocates every agent $\frac{1}{n}\mathbf{d}_i$ of resources such that each agent has a dominant share of $\frac{1}{n}$, which ensures SI.
In step 2, the mechanism repeats the following process till one resource is used up:
Select a set of agents from $G_2$ who have the smallest fraction $t_1$ of resource $r_1$, denoted by $P$, and increase their fractions of resource $r_1$ at the same speed ($\delta^*$) till the fraction reaches the second smallest fraction $t_2$ in $G_2$ ($\delta^*=\delta_0$) or one resource is used up ($\delta^*=\delta_1$ for resource $r_1$ and $\delta^*=\delta_2$ for resource $r_2$).

\myparagraph{Example 1.}
Consider an instance with 3 agents who have demand vectors $\mathbf{d}_1=(1,\frac{2}{5})$, $\mathbf{d}_2=(1,\frac{1}{5})$ and $\mathbf{d}_3=(\frac{1}{5},1)$.
We compare the allocation under $\FOne$ and DRF.
Notice that DRF can also be viewed as a two-step mechanism, where in step 1 every agent gets $\frac{1}{n}$ dominant share (the same as $\FOne$) and in step 2 we increase the dominant share of every agent at the same speed till one resource is used up.
For the above instance, in step 1 all 3 agents get $\frac{1}{3}$ dominant share, and the remaining resource is $C=(\frac{4}{15},\frac{7}{15})$, corresponding to Figure \ref{fig:example_sup1}.
In step 2, under DRF, all agents have the same dominant share $x^*=\frac{1}{\max\{\frac{11}{5},\frac{8}{5}\}}=\frac{5}{11}$ and the final allocation vectors are $A_1=(\frac{5}{11},\frac{2}{11})$, $A_2=(\frac{5}{11},\frac{1}{11})$ and $A_3=(\frac{1}{11},\frac{5}{11})$, corresponding to Figure \ref{fig:example_sup21}.
Under $\FOne$,  we increase the allocation of agent $3$, who currently has the smallest fraction $\frac{1}{15}$ of resource $r_1$, till the second resource $r_2$ is used up and we have $A_3=(\frac{4}{25},\frac{4}{5})$, corresponding to Figure \ref{fig:example_sup22}.
The SW under DRF is $\frac{5}{11} \times 3 \approx 1.36$, while the SW under $\FOne$ is $\frac{1}{3}+\frac{1}{3}+\frac{4}{5} \approx 1.47$.
\medskip

\begin{figure*}[t]
\centering
\subfloat[Step 1\label{fig:example_sup1}]{%
  \includegraphics[scale=0.7]{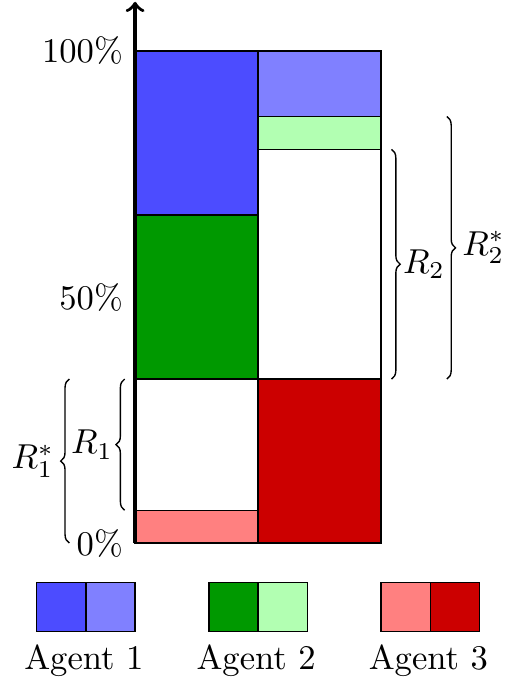}%
}\hspace{3mm}%
\subfloat[DRF Step 2\label{fig:example_sup21}]{%
  \includegraphics[scale=0.7]{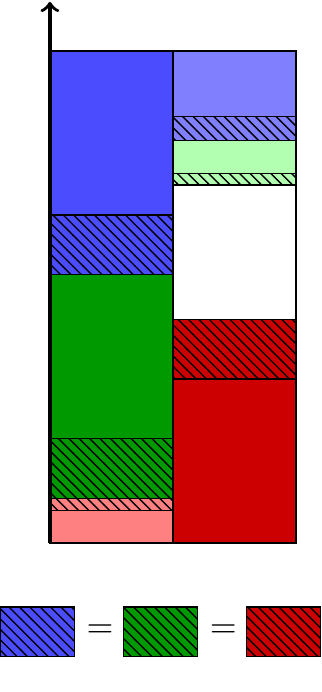}%
}\hspace{4mm}%
\subfloat[$\FOne$ Step 2\label{fig:example_sup22}]{%
  \includegraphics[scale=0.7]{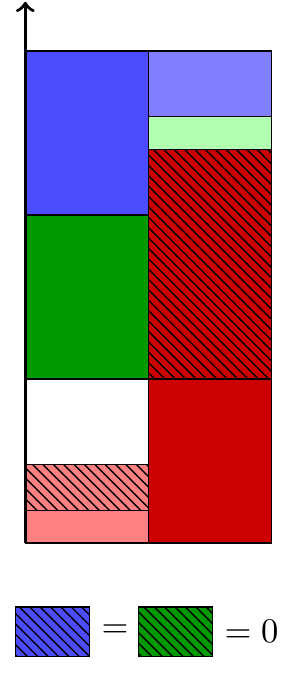}%
}\hspace{4mm}%
\subfloat[$\FTwo$ Step 2\label{fig:example_sup23}]{%
  \includegraphics[scale=0.7]{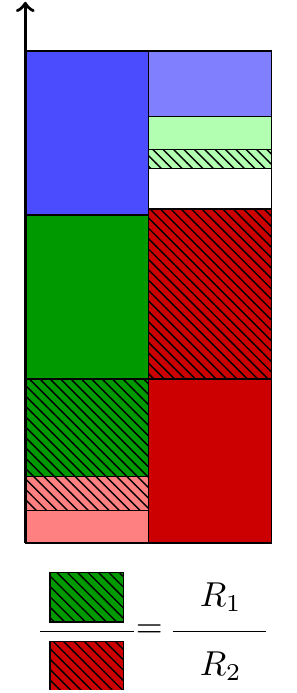}%
}\hspace{4mm}%
\subfloat[$\FTwos$ Step 2\label{fig:example_sup24}]{%
  \includegraphics[scale=0.7]{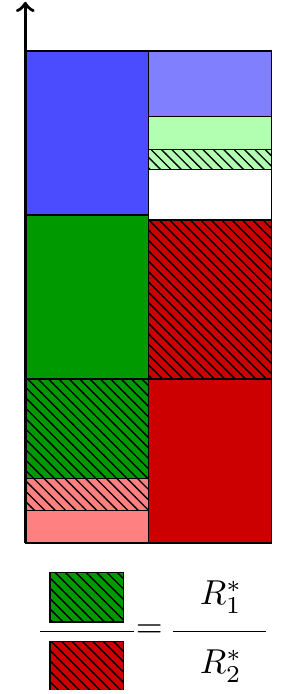}%
}
\caption{Allocations under DRF, $\FOne$, $\FTwo$ and $\FTwos$ in Example 1. The shaded area represents the added parts in respective Step 2.}
\label{fig:example}
\end{figure*}


We show that $\FOne$ satisfies all four properties and has a better \fratio than DRF.

\begin{restatable}{theorem}{FONE}
\label{thm:n-2-F1}
With $2$ resources, mechanism $\FOne$ can be implemented in polynomial time, satisfies SI, EF, PO, and SP, and has
\[\AR_{\SW}(\FOne) = 1 + \alpha \textrm{\quad and \quad} \AR_{\Util}(\FOne) = \frac{1}{1-\alpha}.\]
\end{restatable}

Because $\alpha \in (0, 1/2]$, we have $\AR_{\SW}(\FOne) \leq 3/2$ and $\AR_{\Util}(\FOne) \leq 2$, both of which are significantly better than DRF.

We prove Theorem \ref{thm:n-2-F1} via the following Lemmas \ref{lem:F1-4-properties} to \ref{lem:F1-time}.
We start by showing that $\FOne$ satisfies all four properties.

\begin{lemma}
\label{lem:F1-4-properties}
With $2$ resources, mechanism $\FOne$ satisfies SI, EF, PO, and SP.
\end{lemma}

\begin{proof}
We first show $\FOne$ satisfies SI, EF, and PO.
SI and PO are clearly satisfied since all agents have dominant share at least $\frac{1}{n}$ and the mechanism stops only when one resource is used up.
EF is satisfied in the first step when all agents get the same dominant share $\frac{1}{n}$.
After that, the mechanism allocates more resources to agents in $G_2$, so there is no envy from $G_2$ to $G_1$.
Note that the mechanism stops before any agent in $G_2$ receiving more than $\frac{1}{n}$ of resource~1.
Since all agents in $G_1$ have $\frac{1}{n}$ of resource~1, there is no envy from $G_1$ to $G_2$.
Within $G_2$, the mechanism only allocate resources to agents who have the smallest fraction of resource~1, so no envy will occur.

It remains to show SP is satisfied.
It is easy to see that no agent has an incentive to change the group they belong to as in the final allocation the fraction of the non-dominant resource is at most $\frac{1}{n}$ for all agents in both $G_1$ and $G_2$.
In addition, agents in $G_1$ always get $\frac{1}{n}$ of resource~1, so agents in $G_1$ have no incentive to lie.
Finally, we consider agents in $G_2$.
Suppose that there exists an agent $i_0 \in G_2$ who can benefit by reporting a false demand vector.
Let $\mathbf{A}$ be the truthful allocation and $\mathbf{A}'$ be the manipulated allocation.
Denote $x=\min_{i \in G_2} A_{i,1}$ and $x'=\min_{i \in G_2} A'_{i,1}$.
Since agent $i_0$ gets more utility from $\mathbf{A}'$ than from $\mathbf{A}$, we have $A_{i_0,1} < A'_{i_0,1}$.
Since agent $i_0$ receives more dominant resource in the manipulated outcome, $i_0$ must be one of the agents who have the smallest fraction of resource~1 in the manipulated outcome.
Thereby, $A'_{i_0,1}=x'$ and then $x \leq A_{i_0,1} < A'_{i_0,1} =x'$.
However, this means that agent $i_0$ receives more of both resources in the manipulated outcome while all other agents receive at least the same amount of resources, which contradicts that the truthful allocation is PO.
This finishes the proof for SP.
\end{proof}

Next we prove the claimed upper bounds in Theorem \ref{thm:n-2-F1}.
We show that under $\FTwo$ if resource~1 is used up then $\FTwo$ maximizes both SW and utilization, and if resource~2 is used up then the space left for improvement is at most $\alpha$ of resource~1.

\begin{lemma}
\label{lem:F1-upper}
With $2$ resources,
\[\AR_{\SW}(\FOne) \le 1 + \alpha \textrm{\quad and \quad} \AR_{\Util}(\FOne) \le \frac{1}{1-\alpha}.\]
\end{lemma}

\begin{proof}
Denote by $\mathbf{A}$ the allocation under $\FOne$.
We distinguish two cases depending on which resource is used up in $\mathbf{A}$.

\paragraph{Case 1: resource~1 is used up.}
For this case, we show that $\mathbf{A}$ maximizes both SW and utilization subject to SI and EF.
First, we claim that every agent receives $\frac{1}{n}$ of resource~1 in $\mathbf{A}$.
Suppose this is not case, then there exists some agent~$i^*$ from $G_2$ who receives more than $\frac{1}{n}$ of resource~1 and also more than $\frac{1}{n}$ of resource~2 due to $d_{i^*,1} < d_{i^*,2}$.
Consequently, agent~$i^*$ will be envied by all agents in $G_1$ who receive at most $\frac{1}{n}$ of resource~1 and resource~2, a contradiction to EF.

For utilization, notice that the most efficient way to maximize the use of resource~2 by using a fixed amount of resource~1 subject to SI and EF is to evenly distribute resource~1 among all agents, which is exactly the case in $\mathbf{A}$.
Therefore, allocation $\mathbf{A}$ maximizes utilization subject to SI and EF.

For SW, let $\mathbf{A}'$ be an allocation satisfying SI and EF, and $S_1^{\prime}$ be the sum of resource~1 received by agents in $G_1$ in $\mathbf{A}'$.
Since $\mathbf{A}'$ satisfying SI, we have $S_1^{\prime} \ge \frac{n_1}{n}$.
If $S_1^{\prime} = \frac{n_1}{n}$, then all agents in $G_2$ can receive at most $\frac{n_2}{n}$ of resource~1 totally while the most efficient way to maximize the use of resource~2 is exactly the case in $\mathbf{A}$, so we have $\text{SW}(\mathbf{A}') \le \text{SW}(\mathbf{A})$.
If $S_1^{\prime} > \frac{n_1}{n}$, then the sum of dominant shares of agents in $G_1$ is increased by $S_1^{\prime} - \frac{n_1}{n}$ compared with $\mathbf{A}$, while the sum of dominant shares of agents in $G_2$ is decreased by more than $S_1^{\prime} - \frac{n_1}{n}$ since the demand vectors of agents $i\in G_2$ satisfy that $d_{i,1} < d_{i,2}$.
Thus $\text{SW}(\mathbf{A}') \le \text{SW}(\mathbf{A})$.
Therefore, allocation $\mathbf{A}$ also maximizes SW subject to SI and EF.

\paragraph{Case 2: resource~2 is used up.}
For utilization, since resource~2 is used up and $1-\alpha$ of resource~1 is used by agents in $G_1$ in $\mathbf{A}$, we have $\AR_{\Util}(\FOne) \le \frac{1}{1-\alpha}$.
For SW, let $S_2$ be the sum of dominant shares of agents in $G_2$ in $\mathbf{A}$.
Then $\text{SW}(\mathbf{A})=(1-\alpha)+S_2$ and $S_2 \ge \alpha$.
In any allocation~$\mathbf{A}'$ satisfying SI, all agents in $G_1$ receive at least the same amount of resource~2 as they receive in $\mathbf{A}$,
so the sum of dominant shares of agents in $G_2$ is at most $S_2$ and hence
$\SW(\mathbf{A}') \le 1+S_2$.
Then we have
\[
\AR_{\SW}(\FOne) \le 
\frac{1+S_2}{(1-\alpha)+S_2} \le \frac{1+\alpha}{(1-\alpha)+\alpha} =1+\alpha,
\]
where the second inequality follows from $S_2 \ge \alpha$.
\end{proof}

Then we show the corresponding lower bounds.
The idea is to build an instance where after step 1 the remaining resource is $C=(\alpha-\epsilon,\epsilon)$ with $\epsilon \to 0$. In step 2 $\FOne$ can only increase the allocations of agents in $G_2$ and get SW at most $1+\epsilon$, while the optimal allocation can increase the allocation of one agent in $G_1$ with the smallest $d_{i,2}$ and get SW of $1+\alpha-\epsilon$.

\begin{lemma}
\label{lem:F1-lower}
With $2$ resources,
\[\AR_{\SW}(\FOne) \ge 1 + \alpha \textrm{\quad and \quad} \AR_{\Util}(\FOne) \ge \frac{1}{1-\alpha}.\]
\end{lemma}

\begin{proof}
We first consider SW.
We build an instance with \mpr $\alpha$ and $n$ agents as follows.
The first group $G_1$ consists of $n(1-\alpha)$ agents, where except for one special agent $i^*$ whose demand vector is $(1,\varepsilon)$, all other agents have the same demand vector $(1,1-\varepsilon)$, where $\varepsilon=\frac{1}{n}$.
The second group $G_2$ consists of $n\alpha$ agents who have the same demand vector $(\varepsilon,1)$.
We choose $n$ large enough such that $n(1-\alpha)\geq 2$ and $n\alpha\ge 1$.
Under $\FOne$, in step 2 we will increase the allocations of agents in $G_2$.
However, the remaining amount of resource~2 after step 1 is only $1-\alpha-(n(1-\alpha)-1)\frac{1-\varepsilon}{n}-\frac{\varepsilon}{n} \le \frac{2}{n}$.
So the SW under $\FOne$ is upper bounded by $1+\frac{2}{n}$.
On the other hand, if we give $\frac{1}{n}$ dominant share to every agent except for the special agent $i^*$ and give agent $i^*$ the bundle $(\frac{n-1}{n}\alpha+\frac{1}{n},\frac{n-1}{n^2}\alpha+\frac{1}{n^2})$ such that resource~1 is used up, then the SW is $1+\frac{n-1}{n}\alpha$.
It is easy to verify that the above allocation, denoted by $\mathbf{A}^*$, satisfies SI and EF.
For EF, notice that the first agent receives $\frac{n-1}{n^2}\alpha+\frac{1}{n^2} \le \frac{n+1}{2n^2}\le \frac{n-1}{n^2}$ of resource~2, where the last equality holds as $n\ge 3$, while all other agents receive at least $\frac{1-\varepsilon}{n}=\frac{n-1}{n^2}$ of resource~2.
Then we have 
\[
\AR_{\SW}(\FOne) \ge
\frac{1+\frac{n-1}{n}\alpha}{1+\frac{2}{n}}\overset{n \to \infty}{\longrightarrow} 1+\alpha.
\]

For utilization, we use the same instance as above.
Under $\FOne$, since the SW under $\FOne$ is upper bounded by $1+\frac{2}{n}$, at most $1-\alpha+\frac{\alpha}{n}+\frac{2}{n} \le 1-\alpha+\frac{3}{n}$ of resource~1 is used.
In $\mathbf{A}^*$, resource~1 is used up and at least $1-\frac{2}{n}$ of resource~2 is used, so 
\[
\AR_{\Util}(\FOne) \ge
\frac{1-\frac{2}{n}}{1-\alpha+\frac{3}{n}}\overset{n \to \infty}{\longrightarrow} \frac{1}{1-\alpha}.
\]
\end{proof}

Finally, it is easy to show that $\FOne$ can be implemented in polynomial time.

\begin{lemma}
\label{lem:F1-time}
With $2$ resources, $\FOne$ can be implemented in $O(n^2)$ time.
\end{lemma}

\begin{proof}
Notice that $|P|$ is increasing in each round of step 2, so the number of rounds of step 2 is at most $n$.
Since each round of step 2 can be implemented in $O(n)$ time, $\FOne$ can be implemented in $O(n^2)$ time.
\end{proof}

\subsection{Mechanism $\FTwo$ and $\FTwos$}

According to Theorem \ref{thm:n-2-F1}, $\FOne$ has the worst performance when the population of two groups are balanced, i.e., when $\alpha$ is close to $\frac{1}{2}$, because in step 2 it only increases allocations of agents in one group ($G_2$).
In this case, a better strategy in step 2 is to increase allocations of agents from both groups.


\begin{algorithm}[t]
\DontPrintSemicolon
$C\leftarrow (c_1, c_2) = (1,1)$  \tcp*{remaining resources}
$G_1\leftarrow\{i \mid d_{i,1}=1\}$; $G_2\leftarrow\{i \mid d_{i,1}<1\}$ \\
\ForEach{$i \in N$}
{
  $\mathbf{A}_i\leftarrow\frac{1}{n} \mathbf{d}_i$ \tcp*{every agent receives $\frac{1}{n}$ dominant share}
  $C \leftarrow C-\mathbf{A}_i$
}
$(R_1,R_2) \leftarrow C$ \tcp*{remaining resources after step 1}
\While{$c_1>0$ and $c_2>0$}
{
  $P_1 \leftarrow \arg\min_{i \in G_1} A_{i,2}$ \\
  $P_2 \leftarrow \arg\min_{i \in G_2} A_{i,1}$ \\
  $(\delta^*_1,\delta^*_2) \leftarrow$ CalcStep $()$ \tcp*{calculate increasing steps}
  \ForEach{$k=1,2$}
  {
    \ForEach{$i \in P_k$}
    {
      $\mathbf{A}_i \leftarrow \mathbf{A}_i+ \frac{\delta_k^*}{d_{i,3-k}}\mathbf{d}_i$
      \tcp*{increase the non-dominant resource by the same $\delta_k^*$}
      $C \leftarrow C-\frac{\delta_k^*}{d_{i,3-k}}\mathbf{d}_i$
    }
  }
}
\Return $\mathbf{A}$
\caption{$\FTwo(\mathbf{d}_1,\mathbf{d}_2,\dots,\mathbf{d}_n)$}
\label{alg:F2}
\end{algorithm}

\begin{algorithm}[t]
    \DontPrintSemicolon

      \ForEach{$k = 1,2$}
      {
        $\delta_k \leftarrow \min\limits_{i \in N \setminus P_k} A_{i,3-k}-\min\limits_{i \in P_k} A_{i,3-k}$ \tcp*{increasing step when 2nd smallest fraction of resource $r_{3-k}$ is reached}
        $D_k \leftarrow \sum_{i \in P_k} \frac{1}{d_{i,3-k}}$ \\
        $\overline{\delta_k} \leftarrow \frac{c_{3-k}}{|P_k|+ D_k\frac{R_{3-k}}{R_k}}$
        \tcp*{increasing step when resource $r_{3-k}$ is used up}
        $\delta_k^* \leftarrow \min\{\delta_k,\overline{\delta_k}\}$
      }
      \If{$\frac{\delta_1^* D_1}{\delta_2^* D_2} \le \frac{R_1}{R_2}$}
      {
        $\delta_2^* \leftarrow \delta_1^* \cdot \frac{ D_1}{D_2} \cdot \frac{R_2}{R_1}$ \tcp*{decrease $\delta_2^*$ according to $\delta_1^*$}
      }
      \Else
      {
        $\delta_1^* \leftarrow \delta_2^* \cdot \frac{D_2}{D_1} \cdot \frac{R_1}{R_2}$ \tcp*{decrease $\delta_1^*$ according to $\delta_2^*$}
      }
      \Return $(\delta^*_1,\delta^*_2)$
    \caption{CalcStep $()$}
    \label{alg:F2-step}
    \end{algorithm}

Following this intuition, we propose mechanism $\FTwo$, described in Algorithm \ref{alg:F2}.
Mechanism $\FTwo$ also has two steps. Step 1 is the same as $\FOne$, where every agent gets $\frac{1}{n}$ dominant share.
In step 2, the mechanism increases allocations of agents from both groups, and within each group the method is the same as in $\FOne$, that is, within each group, only agents who have the smallest amount of the non-dominant resource will be allocated more resources, and they will be allocated the same fraction ($\delta_1^*$ or $\delta_2^*$) of the non-dominant resource.
In addition, $\FTwo$ controls the relative allocation rates $(\delta^*_1,\delta^*_2)$ of two groups such that the ratio between the increased dominant shares of two groups is proportional to the ratio between the remaining amounts of two resources after step 1.
Formally, let $\Delta S_1$ and $\Delta S_2$ be the sum of increased dominant share of agents in $G_1$ and $G_2$ in step 2 respectively.
Let $R_1=1-\frac{n_1}{n}-\frac{1}{n}\sum_{i \in G_2}d_{i,1}$ and $R_2=1-\frac{n_2}{n}-\frac{1}{n}\sum_{i \in G_1}d_{i,2}$ be the amount of remaining resources after step 1,
then $\FTwo$ ensures that
\begin{equation}
\label{eq:F2}
\frac{\Delta S_1}{\Delta S_2}=\frac{R_1}{R_2}.
\end{equation}
This condition is crucial to guarantee the good performance of $\FTwo$.

To compute the increasing steps $(\delta^*_1,\delta^*_2)$ (CalcStep $()$ in line 10 of Algorithm \ref{alg:F2}),
we calculate the largest increasing steps $(\delta^*_1,\delta^*_2)$ such that condition (\ref{eq:F2}) is satisfied and one of the following conditions is satisfied:
(a) one resource has been used up; (b) one agent has to be added into $P_1$ or $P_2$.
The concrete algorithm is given in Algorithm \ref{alg:F2-step}.

We now show that $\FTwo$ satisfies SI, EF, and PO, and its \fratio for SW is at most $\frac{4}{3}$.



\begin{restatable}{theorem}{FTWO}
\label{thm:n-2-F2}
With $2$ resources, mechanism $\FTwo$ can be implemented in polynomial time, satisfies SI, EF, and PO, and has
\[\AR_{\SW}(\FTwo) = \frac{4-2\alpha}{3-\alpha} \textrm{\quad and \quad} \AR_{\Util}(\FTwo) = \frac{2}{1+\alpha}.\]
\end{restatable}

We prove the claimed upper bounds via Lemmas \ref{lem:F2-one-side} to \ref{lem:F2-upper-uti}, and the running time and the three properties in Lemma \ref{lem:F2-time-property}. 
The proof of the claimed lower bounds is deferred to the proof of Theorem \ref{thm:n-2-F2s} (Lemma \ref{lem:F2-lower}), where we show the result for both $\FTwo$ and $\FTwos$.

If $R_1=0$ or $R_2=0$, then one resource is used up when all agents have the same dominant share $\frac{1}{n}$ and the allocation of $\FTwo$ is already optimal subject to SI.
In the following analysis we assume $R_1>0$ and $R_2>0$, which also implies that $\Delta S_1>0$ and $\Delta S_2>0$.

We start with the upper bound~$\AR_{\SW}(\FTwo) \le \frac{4-2\alpha}{3-\alpha}$.
Denote the allocation under $\FTwo$ by $\mathbf{A}$ and the SW-maximizing
allocation satisfying SI and EF by $\mathbf{A}^*$.
We can interpret $\mathbf{A}^*$ as the result of a two-step mechanism (similar to the two steps of $\FTwo$), where in step~1 every agent receives $\frac{1}{n}$ dominant share, which is the same as step~1 of $\FTwo$, and the remaining resources are allocated in step~2.
Accordingly, we denote the sum of dominant shares of agents in $G_1$ and $G_2$ in $\mathbf{A}^*$ by $(1-\alpha)+\Delta S^*_1$ and $\alpha+\Delta S^*_2$, respectively.
We then provide two lemmas that compare $\Delta S_k$ and $\Delta S^*_k$ for $k \in \{1,2\}$.



The first lemma shows that for the group whose dominant resource is used up in $A$, agents in that group do not receive fewer dominant resource in $\mathbf{A}$ than in $\mathbf{A}^*$.
Intuitively, suppose resource~1 is used up in $\mathbf{A}$, then the only way to improve the SW is to increase the allocations for $G_2$ and decrease that for $G_1$.
Since $d_{i,1}<d_{i,2}=1$ for all $i \in G_2$, the increase in the dominant share in $G_2$ will outweigh the decrease in the dominant share in $G_1$, thus improving the SW.

\begin{lemma}
\label{lem:F2-one-side}
For any $k \in \{1,2\}$, if resource $k$ is used up in $\mathbf{A}$, then $\Delta S_k \ge \Delta S^*_k$.
\end{lemma}

\begin{proof}
Due to symmetry, it suffices to show the result for $k=1$.
Suppose that resource~1 is used up in $\mathbf{A}$.
Denote the sum of resource~1 received by agents in $G_2$ in $\mathbf{A}$ and $\mathbf{A}^*$ by $x$ and $x^*$ respectively.
Notice that at step 2 of $\FTwo$ we always increase allocations of agents in $G_2$ with the smallest $A_{i,1}$, which is the most efficient way of using $x$ of resource~1 to maximize the received amount of resource~2 subject to EF.
Therefore, for any other allocation satisfying SI and EF within $G_2$, if agents in $G_2$ receive at most $x$ of resource~1, then they can receive at most $\alpha+\Delta S_2$ of resource~2.
We use this property to show that $\Delta S_1 \ge \Delta S^*_1$.

Suppose towards a contradiction that $\Delta S^*_1 > \Delta S_1$ and let $\delta=\Delta S^*_1 - \Delta S_1>0$.
We prepare some inequalities.
First, since resource $1$ is used up in $\mathbf{A}$, we have that
\begin{equation}
\label{eq:F2-ratio-1}
x^* \le x-\delta.
\end{equation}
Next, since $d_{i,1}<1$ for some $i \in G_2$ from the assumption $R_1>0$ and $R_2>0$, in $\mathbf{A}^*$ the amount of resource~1 received by agents in $G_2$ is less than the amount of resource~2, that is,
\begin{equation}
\label{eq:F2-ratio-2}
x^*< \alpha+\Delta S^*_2.
\end{equation}
Finally, since $\text{SW}(\mathbf{A}^*)=1+\Delta S^*_1+\Delta S^*_2 \ge \text{SW}(\mathbf{A})=1+\Delta S_1+\Delta S_2$,
we have $\Delta S_2-\Delta S^*_2 \le \Delta S^*_1-\Delta S_1=\delta$, that is,
\begin{equation}
\label{eq:F2-ratio-3}
\Delta S^*_2+\delta \ge \Delta S_2.
\end{equation}

Now, based on $\mathbf{A}^*$, we construct a partial allocation $\mathbf{A}'$ for $G_2$ by multiplying each $\mathbf{A}^*_i$ by $\gamma$ for all $i \in G_2$, where $\gamma=\frac{x}{x^*}$.
In $\mathbf{A}'$ every agent in $G_2$ has at least $\frac{1}{n}$ dominant share since $\mathbf{A}^*$ satisfies SI and the allocations of agents in $G_2$ are increased from $\mathbf{A}^*$ to $\mathbf{A}'$.
EF is also satisfied within $G_2$ since $\mathbf{A}^*$ satisfies EF and the allocations of agents in $G_2$ are multiplied by the same factor $\gamma$  to get $\mathbf{A}'$.
However, in $\mathbf{A}'$ the sum of resource~2 received by agents in $G_2$ is
\begin{equation*}
\begin{aligned}
\gamma (\alpha+\Delta S^*_2) =\frac{x}{x^*} (\alpha+\Delta S^*_2) & \overset{(\ref{eq:F2-ratio-1})}{\ge} \alpha+\Delta S^*_2+ \frac{\delta}{x^*} (\alpha+\Delta S^*_2) \\
& \overset{(\ref{eq:F2-ratio-2})}{>} \alpha+ \Delta S^*_2 +\delta \\ 
& \overset{(\ref{eq:F2-ratio-3})}{\ge} \alpha+\Delta S_2.
\end{aligned}
\end{equation*}
Then, in $\mathbf{A}'$ agents in $G_2$ have $x$ of resource~1 but more than $\alpha+\Delta S_2$ of resource~2, which is a contradiction.
Therefore, $\Delta S_1 \ge \Delta S^*_1$.
\end{proof}

The second lemma provides a lower bound for the other group whose dominant resource is not used up in $\mathbf{A}$.
For this group, we show that the increased dominant share in step 2 is no less than half of the increased dominant share in step 2 of any allocation satisfying SI and EF.
This lemma will also be used in the analysis of utilization.

\begin{lemma}
\label{lem:F2-other-side}
Let $\mathbf{A}'$ be an arbitrary allocation satisfying SI and EF.
Denote the sum of dominant shares of agents in $G_1$ and $G_2$ in $\mathbf{A}'$ by $(1-\alpha)+\Delta S'_1$ and $\alpha+\Delta S'_2$, respectively.
Let $r_1 \in \{1,2\}$ be the resource that is used up in $\mathbf{A}$ and let $r_2=3-r_1$ represents the other,
then $\Delta S_{r_2} \ge \frac{1}{2} \Delta S'_{r_2}$.
\end{lemma}

\begin{proof}
Let $\beta=\frac{\Delta S_1}{R_1}=\frac{\Delta S_2}{R_2}$.
If $\beta \ge \frac{1}{2}$, then
$\Delta S_i \ge \frac{1}{2} R_i \ge \frac{1}{2} \Delta S'_i$ for any $i\in\{1,2\}$.
Then it suffices to show $\Delta S_{r_2} \ge \frac{1}{2} \Delta S'_{r_2}$ when $\beta < \frac{1}{2}$.
Since $\mathbf{A}'$ satisfies SI, we have $\Delta S'_i \ge 0$ for each $i \in \{1,2\}$.
Similar as $\mathbf{A}^*$, we can interpret $\mathbf{A}'$ as the result of a two-step mechanism, where in step~1 every agent receives $\frac{1}{n}$ dominant share and in step~2 every agent receives its remaining resources in $\mathbf{A}'$.
Note that in step~2 of $\mathbf{A}'$ agents in $G_{r_2}$ can receive at most $R_{r_1}$ of resource~$r_1$.
Our goal is to show that when receiving at most $R_{r_1}$ of resource~$r_1$, they can receive at most $2\Delta S_{r_2}$ of resource~$r_2$.

To maximize the utilization of resource~$r_2$ by using a fixed amount of resource~$r_1$ for agents in $G_{r_2}$ subject to SI and EF, we should first give every agent $\frac{1}{n}$ dominant share and then always increase the share of agents in $G_{r_2}$ with the smallest fraction of resource~$r_1$, which is exactly what $\FTwo$ does.
Observe that the average $d_{i,r_1}$ of agents $i \in G_{r_2}$ with the smallest $A_{i,r_1}$ (i.e., $P_{r_2}$ in Algorithm \ref{alg:F2}) is increasing in step 2 of $\FTwo$, or equivalently $D_{r_2}$ in Algorithm \ref{alg:F2-step} is decreasing.
Moreover, in step 2 of $\FTwo$ agents in $G_{r_2}$ use more than $\frac{1}{2}R_{r_1}$ of resource $r_1$ since resource $r_1$ is used up and $\Delta S_{r_1} < \frac{1}{2}R_{r_1}$.
So even if we allocate all $R_{r_1}$ of resource $r_1$ to agents in $G_{r_2}$ in step~2 of $\FTwo$, they can use at most $2\Delta S_{r_2}$ of resource $r_2$ subject to EF.
Thus $\Delta S'_{r_2} \le 2\Delta S_{r_2}$.
\end{proof}

We are ready to show the claimed upper bound for SW.
\begin{lemma}
\label{lem:F2-upper-SW}
With $2$ resources, $\AR_{\SW}(\FTwo) \le \frac{4-2\alpha}{3-\alpha}$.
\end{lemma}

\begin{proof}
Let $r_1 \in \{1,2\}$ be the resource that is used up in $\mathbf{A}$ and let $r_2=3-r_1$ represents the other.
According to Lemma \ref{lem:F2-one-side} and Lemma \ref{lem:F2-other-side}, we have $\Delta S_{r_1} \ge \Delta S_{r_1}^*$ and $\Delta S_{r_2} \ge \frac{1}{2}\Delta S_{r_2}^*$.
Then we have
\begin{equation*}
\begin{aligned}
\AR_{\SW}(\FTwo) 
\le \frac{1+\Delta S^*_1+\Delta S^*_2}{1+\Delta S_1+\Delta S_2} 
\le \frac{1+\Delta S^*_{r_2}}{1+\Delta S_{r_2}} 
\le \frac{1+\Delta S^*_{r_2}}{1+\frac{1}{2}\Delta S^*_{r_2}} 
\le \frac{2-\alpha}{1+\frac{1-\alpha}{2}} 
= \frac{4-2\alpha}{3-\alpha},
\end{aligned}
\end{equation*}
where the last inequality is due to $\max\{\Delta S^*_1, \Delta S^*_2\} \le \max\{R_1,R_2\} \le 1-\alpha$.
\end{proof}

Then we show the claimed upper bound for utilization, where we need to focus on the resource not fully used.
Lemma \ref{lem:F2-other-side} already provides a lower bound for one group on that resource.
The main task in the following lemma is to show a similar lower bound for the other group on that resource.

\begin{lemma}
\label{lem:F2-upper-uti}
With $2$ resources, $\AR_{\Util}(\FTwo) \le \frac{2}{1+\alpha}$.
\end{lemma}

\begin{proof}
Let $\beta=\frac{\Delta S_1}{R_1}=\frac{\Delta S_2}{R_2}$.
We distinguish two cases when $\beta \ge \frac{1}{2}$ and when $\beta < \frac{1}{2}$.
For the first case, we have that
$\Delta S_i \ge \frac{1}{2} R_i $ for $i=1,2$.
Then at least $1-\frac{R_1}{2} \ge 1-\frac{\alpha}{2} \ge 1-\frac{1-\alpha}{2}$ of resource~1 is used and at least $1-\frac{R_2}{2} \ge 1-\frac{1-\alpha}{2}$ of resource~2 is used.
Therefore, 
\[\AR_{\Util}(\FTwo) \le \frac{1}{1-\frac{1-\alpha}{2}}=\frac{2}{1+\alpha}.\]

For the second case when $\beta < \frac{1}{2}$, we further distinguish two cases depending on which resource is used up in $\mathbf{A}$.
We now use $\mathbf{A}^*$ to represent the utilization-maximizing allocation satisfying SI and EF.
Following the same use of notations, in $\mathbf{A}^*$ agents in $G_1$ and $G_2$ receive $1-\alpha+\Delta S^*_1$ of resource~1 and $\alpha+\Delta S^*_2$ of resource~2 respectively.
Let $y_1=\sum_{i \in G_1} \frac{d_{i,2}}{n}$ be the amount of resource~2 received by agents in $G_1$ and $y_2=\sum_{i \in G_2} \frac{d_{i,1}}{n}$ be the amount of resource~1 received by agents in $G_2$ when every agent receives $\frac{1}{n}$ dominant share.

If resource~1 is used up in $\mathbf{A}$, then since $\mathbf{A^*}$ satisfies SI and EF, according to Lemma \ref{lem:F2-other-side}, we have that $\Delta S^*_2 \le 2\Delta S_2$.
Let $y_1^*$ be the amount of resource~2 received by $G_1$ in $\mathbf{A}^*$.
Next we show $y_1^* \le 2y_1$.
Notice that the most efficient way of using resource~1 to maximize the received amount of resource~2 for $G_1$ subject to EF is to evenly distribute resource~1.
Since agents in $G_1$ use $1-\alpha$ of resource~1 and $y_1$ of resource~2 when every agent gets $\frac{1}{n}$ dominant share and $1-\alpha \ge \frac{1}{2}$, we have $y_1^* \le 2y_1$.
Therefore, we have $\Delta S_2+y_1 \ge \frac{1}{2}(\Delta S^*_2+y_1^*)$ and $\AR_{\Util}(\FTwo)$ (determined by resource~2) is upper bounded by
\[
\frac{\alpha+\Delta S^*_2+y_1^*}{\alpha+\Delta S_2+y_1} 
\le \frac{\alpha+\Delta S^*_2+y_1^*}{\alpha+\frac{1}{2}(\Delta S^*_2+y_1^*)} 
\le \frac{1}{\alpha+\frac{1-\alpha}{2}}
=\frac{2}{1+\alpha},
\]
where the second inequality is due to $\Delta S^*_2+y_1^* \le 1-\alpha$.

If resource~2 is used up in $\mathbf{A}$, then according to Lemma \ref{lem:F2-other-side}, we have that $\Delta S^*_1 \le 2\Delta S_1$.
Let $y_2^*$ be the amount of resource~1 received by $G_2$ in $\mathbf{A}^*$.
Next we show $y_2^* \le 2y_2$.
In step~2 of $\mathbf{A}$, agents in $G_1$ receive more than $\frac{R_2}{2}$ of resource $2$ since $\frac{\Delta S_2}{R_2} = \beta < \frac{1}{2}$ and resource $2$ is used up, while they receive less than $\frac{R_1}{2}$ of their dominant resource $1$ as $\frac{\Delta S_1}{R_1} = \beta < \frac{1}{2}$.
Then from $\frac{R_1}{2} \ge \frac{R_2}{2}$ we have $R_2 \le R_1 \le \alpha$.
Again, notice that the most efficient way of using resource~2 to maximize the received amount of resource~1 for $G_2$ subject to EF is to evenly distribute resource~2.
Since agents in $G_2$ use $y_2$ of resource~$1$ and $\alpha$ of resource~2 when every agent gets $\frac{1}{n}$ dominant share and $R_2 \le \alpha$, we have $y_2^* \le 2y_2$.
Therefore, we have $\Delta S_1+y_2 \ge \frac{1}{2}(\Delta S^*_1+y_2^*)$ and $\AR_{\Util}(\FTwo)$ (determined by resource~1) is upper bounded by
\[
\frac{1-\alpha+\Delta S^*_1+y_2^*}{1-\alpha+\Delta S_1+y_2}
\le \frac{1-\alpha+\Delta S^*_1+y_2^*}{1-\alpha+\frac{1}{2}(\Delta S^*_1+y_2^*)}
\le \frac{1}{1-\frac{\alpha}{2}} 
\le \frac{1}{1-\frac{1-\alpha}{2}}
=\frac{2}{1+\alpha},
\]
where the second inequality is due to $\Delta S^*_1+y_2^* \le \alpha$.
\end{proof}

Finally, we show the running time and properties of $\FTwo$.
\begin{lemma}
\label{lem:F2-time-property}
With $2$ resources, mechanism $\FTwo$ satisfies SI, EF, and PO, and can be implemented in $O(n^2)$ time.
\end{lemma}

\begin{proof}
The proof is very similar to the proof for $\FOne$.
For the properties, SI and PO are clearly satisfied since all agents have dominant share at least $\frac{1}{n}$ and the mechanism stops only when one resource is used up.
EF is satisfied in step 1 as all agents have the same dominant share $\frac{1}{n}$.
Note that the mechanism stops before any agent receiving more than $\frac{1}{n}$ of the non-dominant resource while all agents have at least $\frac{1}{n}$ of the dominant resource, so there is no envy between $G_1$ and $G_2$.
Within each group, in step 2 the mechanism only allocates resources to agents who have the smallest fraction of the non-dominant resource, so no envy will occur.

For the running time, notice that $|P_1 \cup P_2|$ is increasing in each round of step 2, so the number of rounds of step 2 is at most $n$.
Since each round of step 2 can be implemented in $O(n)$ time, $\FTwo$ can be implemented in $O(n^2)$ time.
\end{proof}


However, $\FTwo$ does not satisfy SP as the agent with the minimum $d_{i,1}$ (or minimum $d_{i,2}$) could influence the ratio $\frac{R_1}{R_2}$ by modifying its demand vector to get more resources in step 2, as shown in the following example.

\myparagraph{Example 2.}
Consider an instance with two agents who have demand vectors $\mathbf{d}_1=(1,\frac{1}{2})$ and $\mathbf{d}_2=(\frac{1}{4},1)$.
According to $\FTwo$, in step 1 agent 1 gets $(\frac{1}{2},\frac{1}{4})$ and agent 2 gets $(\frac{1}{8},\frac{1}{2})$.
Then the remaining resources is $(\frac{3}{8},\frac{1}{4})$ and the increasing speed ratio is $\frac{3}{2}$.
In step 2, agent 1 gets $(\frac{3}{14},\frac{3}{28})$ and agent 2 gets $(\frac{1}{28},\frac{1}{7})$, and resource~2 is used up.
Overall agent 2 gets $(\frac{9}{56},\frac{9}{14})$.
However, if agent 2 reports another demand vector $\mathbf{d}'_2=(\frac{1}{2},1)$, then both agents will get the same dominant share $\frac{2}{3}$ under $\FTwo$.
In particular, agent 2 will get $(\frac{1}{3},\frac{2}{3})$, which is strictly better than $(\frac{9}{56},\frac{9}{14})$.
Therefore, $\FTwo$ is not SP.
\medskip

Fortunately, we can make $\FTwo$ satisfy SP with a small modification.
In the following we propose a slightly different mechanism $\FTwos$ that replaces the condition (\ref{eq:F2}) by the following condition:
\begin{equation}
\label{eq:F2star}
\frac{\Delta S_1}{\Delta S_2}=\frac{R^*_1}{R^*_2}=\frac{R_1+\frac{1}{n}d_{i^*,1}}{R_2+\frac{1}{n}d_{j^*,2}},
\end{equation}
where $i^*$ is an agent in $G_2$ with the minimum $d_{i,1}$ and $j^*$ is an agent in $G_1$ with the minimum $d_{i,2}$.
That is, the ratio between $\Delta S_1$ and $\Delta S_2$ is proportional to the ratio between the remaining amounts of two resources when all agents except $i^*$ and $j^*$ get $\frac{1}{n}$ dominant share.
Intuitively, for agent $i^*$, this modification prevents it from increasing $d_{i^*,1}$ to influence $\frac{R^*_1}{R^*_2}$, unless $d_{i^*,1}$ becomes larger than the second smallest $d_{i,1}$, for which case we can show that $i^*$ cannot benefit.

We show that $\FTwos$ satisfies all four properties including SP, and its \fratio is very close to that of $\FTwo$.

\begin{restatable}{theorem}{FTWOS}
\label{thm:n-2-F2s}
With $2$ resources, mechanism $\FTwos$ can be implemented in polynomial time, satisfies SI, EF, PO, and SP, and has
\[\AR_{\SW}(\FTwos) \in \left[\frac{4-2\alpha}{3-\alpha}, \frac{4-2\alpha}{3-\alpha-\frac{1}{n}}\right]
\textrm{\quad{and}\quad}
\AR_{\Util}(\FTwos) \in \left[\frac{2}{1+\alpha}, \frac{2}{1+\alpha-\frac{1}{n}}\right].
\]
\end{restatable}

We prove Theorem \ref{thm:n-2-F2s} via the following Lemmas \ref{lem:F2s-SI,EF,PO} to \ref{lem:F2-lower}.
We first show that $\FTwos$ satisfies SP.

\begin{lemma}
  \label{lem:F2s-SP}
  With $2$ resources, mechanism $\FTwos$ satisfies SP.
  \end{lemma}
  
  \begin{proof}
  For a contradiction, suppose some agent $i_0 \in G_1$, whose true demand vector is $\mathbf{d}_{i_0}=(1,d_{i_0,2})$, reports a false demand vector $\mathbf{d}'_{i_0}$ and reaches a higher utility.
  Notice that under $\FTwos$ every agent can get at least $\frac{1}{n}$ of the dominant resource and at most $\frac{1}{n}$ of the non-dominant resource, so no agent can benefit by changing the group they belong to. 
  Then we can assume that $\mathbf{d}'_{i_0}=(1,d'_{i_0,2})$.
  Let $\mathbf{A}$ be the truthful allocation and let $\mathbf{A}'$ be the manipulated allocation.
  Let $\Delta S'_1$ and $\Delta S'_2$ be the sum of increased dominant shares of agents in $G_1$ and $G_2$, respectively, at step 2 of $\FTwos$ when $i_0$ reports $(1,d'_{i_0,2})$.
  Denote $x=\min_{i \in G_1} A_{i,2}$ and $x'=\min_{i \in G_1} A'_{i,2}$.
  
  We first show $\Delta S'_1 > \Delta S_1$ and $d'_{i_0,2} > d_{i_0,2}$.
  Since agent $i_0$'s utility is increased in $\mathbf{A}'$, we have $A'_{i_0,1}>A_{i_0,1}$ and $A'_{i_0,2}>A_{i_0,2}$.
  From $A'_{i_0,1}>A_{i_0,1}\ge\frac{1}{n}$ we have that agent $i_0$'s allocation is increased in step 2, so $A'_{i_0,2}=x'$.
  By the definition of $x$ we have $A_{i_0,2} \ge x$.
  Putting all these together we have~$x'=A'_{i_0,2}>A_{i_0,2} \ge x$.
  This implies that except $i_0$ all other agents in $G_1$ should receive at least the same amount of both resources.
  Since agent $i_0$'s utility is increased, we conclude that in $\mathbf{A}'$ agents in $G_1$ receive more amount of both resources, and consequently agents in $G_2$ receive less amount of both resources. 
  That is, $\Delta S'_1 > \Delta S_1$ and $\Delta S'_2 < \Delta S_2$, which leads to $\frac{\Delta S'_1}{\Delta S'_2} > \frac{\Delta S_1}{\Delta S_2}$.
  According to the increasing speed condition $\frac{\Delta S_1}{\Delta S_2}=\frac{R^*_1}{R^*_2}$, $\frac{\Delta S'_1}{\Delta S'_2} > \frac{\Delta S_1}{\Delta S_2}$ implies that $R^*_2$ is decreased when agent $i_0$ reports $(1,d'_{i_0,2})$.
  It follows that $d'_{i_0,2} > d_{i_0,2}$.
  In particular, agent $i_0$ is not the agent in $G_1$ with the minimum $d_{i,2}$ when it reports $(1,d'_{i_0,2})$.

  Let $$\delta=\frac{1}{n}(d'_{i_0,2} - \max\{d_{i_0,2},\min_{j\in G_1-i_0} d_{j,2}\}) \ge 0$$ be the influence on $R^*_2$ when agent $i_0$ reports $(1,d'_{i_0,2})$.
  In the following we show two contradictory results $\Delta S'_2 > \Delta S_2 -\delta$ and  $\Delta S'_2 \le \Delta S_2 -\delta$, which finish our proof.
  We start with $\Delta S'_2 > \Delta S_2 -\delta$.
  For the manipulated instance, we have
  \begin{equation}
  \label{eq:F2-SP-1}
  \frac{\Delta S'_1}{\Delta S'_2} 
  =\frac{R^*_1}{R^*_2-\delta}
  \Rightarrow
  \frac{\Delta S'_2}{R^*_2-\delta} = \frac{\Delta S'_1}{R^*_1}.
  \end{equation}
  For the original instance, we have
  \begin{equation}
  \label{eq:F2-SP-2}
  \frac{\Delta S_1}{\Delta S_2} = \frac{R^*_1}{R^*_2}
  \Rightarrow
  \frac{\Delta S_1}{R^*_1} = \frac{\Delta S_2}{R^*_2}.
  \end{equation}
  Combining Equations (\ref{eq:F2-SP-1}) and (\ref{eq:F2-SP-2}), we have
  \begin{equation}
  \label{eq:F2-SP-3}
  \frac{\Delta S'_2}{R^*_2-\delta} = \frac{\Delta S'_1}{R^*_1}
  >\frac{\Delta S_1}{R^*_1} = \frac{\Delta S_2}{R^*_2}
  >\frac{\Delta S_2-\delta}{R^*_2-\delta},
  \end{equation}
  where the first inequality follows by $\Delta S'_1>\Delta S_1$ and the second follows by $\Delta S_2<R_2^*$ and $\delta<R_2^*$.
  From (\ref{eq:F2-SP-3}) we have $\Delta S'_2 > \Delta S_2 -\delta$.

  We continue with $\Delta S'_2 \le \Delta S_2 -\delta$.
  Recall that we have shown $x'=A'_{i_0,2}>A_{i_0,2} \ge x$, and   
  $A'_{j,1} \ge A_{j,1}$ and $A'_{j,2} \ge A_{j,2}$ for all $j \in G_1$.
  In order to show $\Delta S'_2 \le \Delta S_2 -\delta$, we strengthen these two results by showing that $x' \ge x+\delta$, and 
  \begin{equation}
  \label{eq:F2-SP-4}
  \sum_{j \in G_1} (A'_{j,1} - A_{j,1}) \ge \delta \text{ and }
  \sum_{j \in G_1} (A'_{j,2} - A_{j,2}) \ge \delta.
  \end{equation}
  We first show $x' \ge x+\delta$.
  Since agent $i_0$'s utility is increased by reporting $(1,d'_{i_0,2})$, we have $A'_{i_0,1}>A_{i_0,1}\ge \frac{1}{n}$ and
  \[
  A'_{i_0,2} = A'_{i_0,1} \cdot d'_{i_0,2} \ge A'_{i_0,1} \cdot (d_{i_0,2}+n\delta)>A_{i_0,1} \cdot (d_{i_0,2}+n\delta) \ge A_{i_0,2} +\delta.
  \]
  The first inequality is from the definition of $\delta$.
  Then $x'=A'_{i_0,2}\ge A_{i_0,2} +\delta\ge x+\delta$.
  
  Next we show claims in (\ref{eq:F2-SP-4}).
  Let $j^* \in G_1$ be one of the agents $j$ in $G_1$ with the minimum $d_{j,2}$ when $i_0$ reports $(1,d'_{i_0,2})$.
  Recall that $j^* \neq i_0$.
  We now prove $A'_{j^*,2}-A_{j^*,2} \ge \delta$, and then
  it follows that $A'_{j^*,1}-A_{j^*,1} \ge \delta$
  since $d_{j^*,1}=1 \ge d_{j^*,2}$.
  By definition, $j^*$ must be one of the agents in $G_1$ who have the minimum amount of resource~2 in $\mathbf{A}'$, i.e., $A'_{j^*,2}=x'$.
  Then we differentiate two cases when $A_{j^*,2}=x$ and when $A_{j^*,2}>x$ in $\mathbf{A}$ to prove $A'_{j^*,2}-A_{j^*,2} \ge \delta$.
  In the first case, it can be proved by $A'_{j^*,2}=x'\geq x+\delta=A_{j^*,2}+\delta$.
  In the second case, $j^*$ does not receive any resource at step 2, so $A_{j^*,2}=\frac{1}{n}d_{j^*,2}$. 
  By the definition of $\delta$, we have $d'_{i_0,2} \geq d_{j^*,2}+n\delta$.
  Then
  \[
    A'_{j^*,2}=x'=A'_{i_0,2}=A'_{i_0,1}d'_{i_0,2} \geq A'_{i_0,1}(d_{j^*,2}+n\delta) \geq \frac{1}{n}(d_{j^*,2}+n\delta)=A_{j^*,2}+\delta.
  \]
  To sum up, in $\mathbf{A}'$ all agents in $G_1$ receive at least the same amount of resources that they receive in $\mathbf{A}$ and specifically for $j^*$ we have $A'_{j^*,1}- A_{j^*,1} \ge\delta$ and $A'_{j^*,2}-A_{j^*,2} \ge \delta$,
  so we have claims in (\ref{eq:F2-SP-4}).

  Finally, based on claims in (\ref{eq:F2-SP-4}), we show $\Delta S'_2 \le \Delta S_2 -\delta$. 
  If in $\mathbf{A}$ resource~2 is used up, then $\Delta S'_2 \le \Delta S_2 -\delta$ because $\sum_{j \in G_1} (A'_{j,2} - A_{j,2}) \ge \delta$.
  If in $\mathbf{A}$ resource~1 is used up, then according to
  $\sum_{j \in G_1} (A'_{j,1} - A_{j,1}) \ge \delta$, we have
  $\sum_{j \in G_2} A'_{j,1} \le \sum_{j \in G_2} A_{j,1} - \delta$, and hence
  $\Delta S'_2 \le \Delta S_2 -\delta$ as $d_{j,1} \le d_{j,2}=1$ for $j \in G_2$.
  Therefore, we have $\Delta S'_2 \le \Delta S_2 -\delta$, no matter which resource is used up in $\mathbf{A}$.
  This finishes the proof.
  \end{proof}

The proof of other three properties and the running time is the same as $\FTwo$.

\begin{lemma}
\label{lem:F2s-SI,EF,PO}
$\FTwos$ satisfies SI, EF, and PO, and can be implemented in $O(n^2)$ time.
\end{lemma}

\begin{proof}
The proof is the same as the proof for $\FTwo$ (Lemma \ref{lem:F2-time-property}).
\end{proof}

Next, we show that the upper bound for the \fratio{s} of $\FTwos$ is very close to that of $\FTwo$.
The proof is similar to the proof for $\FTwo$ (Lemmas \ref{lem:F2-one-side} to \ref{lem:F2-upper-uti}).

\begin{lemma}
\label{lem:F2star-upper}
With $2$ resources, 
\[\AR_{\SW}(\FTwos) \le \frac{4-2\alpha}{3-\alpha-\frac{1}{n}}
\textrm{\quad{and}\quad}
\AR_{\Util}(\FTwos) \le \frac{2}{1+\alpha-\frac{1}{n}}.\]
\end{lemma}

\begin{proof}
In this proof, denote by $\mathbf{A}$ the allocation of $\FTwos$ and $\mathbf{A}^*$ the SW maximization (or utilization maximization) satisfying SI and EF.
Denote in the step 2 of $\FTwos$ the sum of increased dominant shares of agents in $G_1$ and $G_2$  by $\Delta S_1$ and $\Delta S_2$, respectively. 
Let $R_1$ and $R_2$ be the remaining amount of two resources after the step 1 of $\FTwos$.
Let $R_1^*=R_1+\frac{1}{n}d_{i^*,1}$ and $R_2^*=R_2+\frac{1}{n}d_{j^*,2}$,
where $i^* \in G_2$ is the agent $i$ in $G_2$ with the minimum $d_{i,1}$ and $j^* \in G_1$ is the agent $j$ in $G_1$ with the minimum $d_{j,2}$.
Denote the sum of dominant shares of agents in $G_1$ and $G_2$ in $\mathbf{A}^*$ by $(1-\alpha)+\Delta S^*_1$ and $\alpha+\Delta S^*_2$, respectively.
Let $\beta=\frac{\Delta S_1}{R_1^*}=\frac{\Delta S_2}{R_2^*}$.

We first consider SW.
Let $r_1 \in \{1,2\}$ be the resource that is used up in $\mathbf{A}$ and let $r_2=3-r_1$ represents the other.
Recall that in the proof of Lemma \ref{lem:F2-one-side} we do not need the property $\frac{\Delta S_1}{\Delta S_2}=\frac{R_1}{R_2}$ of $\FTwo$, thus the result in Lemma \ref{lem:F2-one-side} also holds for $\FTwos$, i.e., $\Delta S_{r_1} \ge \Delta S^*_{r_1}$.
For resource~$r_2$, if $\beta \ge \frac{1}{2}$, we have
$\Delta S_{r_2} \ge \frac{1}{2} R^*_{r_2} \ge \frac{1}{2} R_{r_2}\ge \frac{1}{2} \Delta S^*_{r_2}$.
If $\beta < \frac{1}{2}$,
recall that in Lemma \ref{lem:F2-other-side} we show that for $\FTwo$ we have $\Delta S^*_{r_2} \le 2\Delta S_{r_2}$.
However, this does not hold for $\FTwos$.
To get a similar bound for $\FTwos$, we interpret $\FTwos$ as a mechanism consisting of 3 steps:
In step $1'$ all agents except $i^*$ (the agent in $G_{r_2}$ with the minimum $d_{i,{r_1}}$) receive $\frac{1}{n}$ dominant share; In step $2'$ agent $i^*$ receive $\frac{1}{n}$ dominant share; The step $3'$ is the same as the original step 2.
Now we compare $\Delta S_{r_2}+\frac{1}{n}$ and $\Delta S^*_{r_2}+\frac{1}{n}$, where the additional $\frac{1}{n}$ can be imagined as the amount of resource~$r_2$ received by $i^*$ in step $2'$.
Since $\beta < \frac{1}{2}$, in step $3'$ agents in $G_{r_1}$ receive at most $\frac{R_{r_1}^*}{2}$ of resource~$r_1$.
Then, at least $\frac{d_{i^*,1}}{n}+(\frac{R_{r_1}^*}{2}-\frac{d_{i^*,1}}{n})$ of resource~$r_1$ is allocated to agents in $G_{r_2}$ in step $2'$ and step $3'$.
In other words, agents in $G_{r_2}$ receive at least $\frac{R_{r_1}^*}{2}$ of resource~${r_1}$ and $\Delta S_{r_2}+\frac{1}{n}$ of resource~${r_2}$ in step $2'$ and step $3'$.
Then even if we allocate all $R_{r_1}^*$ of resource~${r_1}$ to agents in $G_{r_2}$ in step $2'$ and step $3'$, they can use at most $2(\Delta S_{r_2}+\frac{1}{n})$ of resource~${r_2}$ since $i^* \in G_{r_2}$ is the agent in $G_2$ with the minimum $d_{i,{r_1}}$ and $D_k$ in Algorithm \ref{alg:F2-step} is decreasing.
It follows that 
\[\Delta S^*_{r_2}+\frac{1}{n} \le 2(\Delta S_{r_2}+\frac{1}{n}).\]
Therefore, $\AR_{\SW}(\FTwos)$ is upper bounded by
\begin{equation*}
    \begin{aligned}
\frac{1+\Delta S^*_{r_1}+\Delta S^*_{r_2}}{1+\Delta S_{r_1}+\Delta S_{r_2}}
\le \frac{1+\Delta S^*_{r_2}}{1+\Delta S_{r_2}}
&=\frac{1-\frac{1}{n}+\Delta S^*_{r_2}+\frac{1}{n}}{1-\frac{1}{n}+\Delta S_{r_2}+\frac{1}{n}} \\
&\le \frac{1-\frac{1}{n}+\Delta S^*_{r_2}+\frac{1}{n}}{1-\frac{1}{n}+\frac{\Delta S^*_{r_2}+\frac{1}{n}}{2}} \\
&\le \frac{1-\frac{1}{n}+(1-\alpha+\frac{1}{n})}{1-\frac{1}{n}+\frac{1-\alpha+\frac{1}{n}}{2}} \\
&=\frac{4-2\alpha}{3-\alpha-\frac{1}{n}},
    \end{aligned}
\end{equation*}
where the last inequality is due to $\max\{\Delta S^*_1, \Delta S^*_2\} \le \max\{R_1,R_2\} \le 1-\alpha$.

Next we consider utilization.
We distinguish two cases when $\beta \ge \frac{1}{2}$ and when $\beta < \frac{1}{2}$.
For the first case, we have that
$\Delta S_i \ge \frac{R^*_i}{2}  \ge \frac{R_i}{2} $ for any $i \in \{1,2\}$.
Then at least $1-\frac{R_1}{2} \ge 1-\frac{\alpha}{2} \ge 1-\frac{1-\alpha}{2}$ of resource~1 is used and at least $1-\frac{R_2}{2} \ge 1-\frac{1-\alpha}{2}$ of resource~2 is used.
Therefore, 
\[
\AR_{\Util}(\FTwos) \le \frac{1}{1-\frac{1-\alpha}{2}} \le \frac{2}{1+\alpha-\frac{1}{n}}.
\]

For the second case when $\beta < \frac{1}{2}$, we further distinguish two cases when resource~1 is used up and when resource~2 is used up in $\mathbf{A}$.
Let $y_1=\sum_{i \in G_1} \frac{d_{i,2}}{n}$ be the amount of resource~2 received by agents in $G_1$ and $y_2=\sum_{i \in G_2} \frac{d_{i,1}}{n}$ be the amount of resource~1 received by agents in $G_2$ when every agent receives $\frac{1}{n}$ dominant share.

If resource~1 is used up, as shown in the above we have that $\Delta S^*_2+\frac{1}{n} \le 2(\Delta S_2+\frac{1}{n})$.
Let $y_1^*$ be the amount of resource~2 received by $G_1$ in $\mathbf{A}^*$.
Since $1-\alpha \ge \frac{1}{2}$, as shown in Lemma \ref{lem:F2-upper-uti} we have $y_1^* \le 2y_1$.
Therefore, $\AR_{\Util}(\FTwos)$ (determined by resource~2) is upper bounded by
\begin{equation*}
    \begin{aligned}
\frac{\alpha-\frac{1}{n}+\Delta S^*_2+\frac{1}{n}+y_1^*}{\alpha-\frac{1}{n}+\Delta S_2+\frac{1}{n}+y_1}
&\le \frac{\alpha-\frac{1}{n}+\Delta S^*_2+\frac{1}{n}+y_1^*}{\alpha-\frac{1}{n}+\frac{\Delta S^*_2+\frac{1}{n}+y_1^*}{2}} \\
&\le \frac{\alpha-\frac{1}{n}+1-\alpha+\frac{1}{n}}{\alpha-\frac{1}{n}+\frac{1-\alpha+\frac{1}{n}}{2}}\\
&=\frac{2}{1+\alpha-\frac{1}{n}}.
    \end{aligned}
\end{equation*}

If resource~2 is used up, as shown in the above  we have that $\Delta S^*_1+\frac{1}{n} \le 2(\Delta S_1+\frac{1}{n})$.
If $\alpha=\frac{1}{2}$, the situation is the same as the case when resource~1 is used up. It suffices to show the case when $\alpha<\frac{1}{2}$, which means $(1-\alpha)-\alpha\geq \frac{1}{n}$.
Let $y_2^*$ be the amount of resource~1 received by $G_2$ in $\mathbf{A}^*$.
Reusing the technique in Lemma \ref{lem:F2-upper-uti}, we have $R_2 \le \alpha + \frac{1}{n}$ and $y_2^* \le 2y_2+\frac{1}{n}$.
Therefore, $\AR_{\Util}(\FTwos)$ (determined by resource~1) is upper bounded by
\begin{equation*}
    \begin{aligned}
\frac{1-\alpha-\frac{1}{n}+\Delta S^*_1+\frac{1}{n}+y_2^*}{1-\alpha-\frac{1}{n}+\Delta S_1+\frac{1}{n}+y_2}
&\le \frac{1-\alpha-\frac{1}{n}+\Delta S^*_1+\frac{1}{n}+y_2^*}{1-\alpha-\frac{1}{n}+\frac{\Delta S^*_1+\frac{1}{n}+y_2^*}{2}}\\
&\le \frac{1-\alpha-\frac{1}{n}+\alpha+\frac{1}{n}}{1-\alpha-\frac{3}{2n}+\frac{\alpha+\frac{1}{n}}{2}}\\
&=\frac{1}{1-(\frac{\alpha}{2}+\frac{1}{2n})-\frac{1}{2n}}\\
&\le \frac{1}{1-\frac{1-\alpha}{2}-\frac{1}{2n}} \\
&= \frac{2}{1+\alpha-\frac{1}{n}}.
    \end{aligned}
\end{equation*}
where the last inequality follows from $(1-\alpha)-\alpha\geq \frac{1}{n}$.
\end{proof}

Finally we show the lower bounds for both $\FTwo$ and $\FTwos$.
We will build an instance where after step 1 the remaining resource is $C=(\epsilon,1-\alpha-\epsilon)$ with $\epsilon \to 0$. In step 2 the optimal allocation can allocate all the remaining resource to some agent $i^*$ in $G_2$   and get SW close to~$2-\alpha$ and utilization close to~1, while for $\FTwo$, because of the condition (\ref{eq:F2}), we can only give about half of the remaining resource 1 to agent~$i^*$ and the other half to agents in $G_1$ such that $\frac{\Delta S_1}{R_1}=\frac{\Delta S_2}{R_2}\approx\frac{1}{2}$, where the SW is about $1+\frac{1-\alpha}{2}=\frac{1}{2}(3-\alpha)$ and the utilization is about~$\frac{1+\alpha}{2}$.

\begin{lemma}
\label{lem:F2-lower}
With $2$ resources,
\[\AR_{\SW}(\FTwo) \ge \frac{4-2\alpha}{3-\alpha}, \quad \AR_{\SW}(\FTwos) \ge \frac{4-2\alpha}{3-\alpha},\]
and
\[
\AR_{\Util}(\FTwo) \ge \frac{2}{1+\alpha}, \quad \AR_{\Util}(\FTwos) \ge \frac{2}{1+\alpha}.\]
\end{lemma}

\begin{proof}
We first study SW.
We build an instance with \mpr $\alpha$ and $n$ agents as follows.
The first group $G_1$ consists of $n(1-\alpha)$ agents who have the same demand vector $(1,\varepsilon)$, where $\varepsilon=\frac{1}{n^2}$.
The second group $G_2$ consists of $n\alpha$ agents, where except for one special agent $i^*$ whose demand vector is $(\frac{1}{n(1-\alpha)},1)$, all other agents have the same demand vector $(1-\varepsilon,1)$.
The idea is that under $\FTwo$ and $\FTwos$ the special agent $i^*$ can get only about $\frac{1-\alpha}{2}$ of resource~2 and the SW is about $1+\frac{1-\alpha}{2}$ while there exists an allocation that satisfies SI and EF, and has SW about $2-\alpha-\frac{2}{n}$ by giving roughly $1-\alpha$ dominant share to the special agent $i^*$.

Formally, we first give $\frac{1}{n}$ dominant share to every agent except for agent $i^*$.
Then the remaining amount of two resources are $R_1^0=\alpha-(n\alpha-1)\frac{1-\frac{1}{n^2}}{n}\ge \frac{1}{n}$ and $R_2^0=(1-\alpha)(1-\frac{1}{n^2})+\frac{1}{n} \ge 1-\alpha$.
 We can give agent $i^*$ the bundle $(\frac{1-\frac{1}{n^2}}{n},(1-\frac{1}{n^2})(1-\alpha))$ and allocate remaining resources evenly to agents in $G_1$.
The SW is lower bounded by
\[
1-\frac{1}{n}+(1-\frac{1}{n^2})(1-\alpha) \ge 2-\alpha-\frac{2}{n}.
\]
It is easy to verify that the above allocation, denoted by $\mathbf{A}^*$, satisfies SI and EF.

Under $\FTwo$, the remaining resources after step 1 are
\[
\begin{aligned}
&R_1=\alpha-\frac{\frac{1}{n(1-\alpha)}+(n\alpha-1)(1-\frac{1}{n^2})}{n} < \frac{1}{n} \\
& R_2=(1-\alpha)(1-\frac{1}{n^2}) < 1-\alpha.
\end{aligned}
\]
If in step 2 we give the special agent $i^*$ a bundle $(\frac{1}{2n},\frac{1-\alpha}{2})$, then the increased dominant shares for two groups $\Delta S_1$ and $\Delta S_2$ satisfy that
$\Delta S_1 \le R_1-\frac{1}{2n}<\frac{R_1}{2}$ and $\Delta S_2=\frac{1-\alpha}{2}>\frac{R_2}{2}$.
This means in step 2 the dominant share of the special agent $i^*$ is increased by at most $\frac{1-\alpha}{2}$. Then the SW under $\FTwo$ is upper bounded by $1+R_1+\frac{1-\alpha}{2} \le 1+\frac{1-\alpha}{2}+\frac{1}{n}$ and we have
\[
\AR_{\SW}(\FTwo) \ge 
\frac{2-\alpha-\frac{2}{n}}{1+\frac{1-\alpha}{2}+\frac{1}{n}} \overset{n \to \infty}{\longrightarrow} \frac{2-\alpha}{1+\frac{1-\alpha}{2}}.
\]

For $\FTwos$, we have
\[
\frac{R^*_1}{R^*_2}=\frac{R_1+\frac{1}{n^2(1-\alpha)}}{R_2+\frac{1}{n^2}} \ge \frac{R_1+\frac{1}{n^2}}{R_2+\frac{1}{n^2}} \ge \frac{R_1}{R_2},
\]
where the last inequality follows by $R_1<R_2$.
This means that under $\FTwos$ the special agent gets less resources in step 2 than under $\FTwo$, i.e., $\Delta S^*_2 \le \Delta S_2$.
Using the same argument for $\FTwo$ we get that~$\AR_{\SW}(\FTwos)$ is also lower bounded by $\frac{2-\alpha}{1+\frac{1-\alpha}{2}}$.

For utilization, we use the same instance in the above.
In $\mathbf{A}^*$, when we give agent $i^*$ the bundle $(\frac{1-\frac{1}{n^2}}{n},(1-\frac{1}{n^2})(1-\alpha))$ and every other agent $\frac{1}{n}$ dominant share, the remaining amount of resource~1 is at most $R_1^0 \le \frac{\alpha}{n^2}+\frac{1}{n} \le \frac{2}{n}$ and the remaining amount of resource~2 is at most $R_2^0-(1-\frac{1}{n^2})(1-\alpha)=\frac{1}{n}$.
Thus, utilization of $\mathbf{A}^*$ is at least $1-\frac{2}{n}$.

Under $\FTwo$, the remaining resources after step 1 are $R_1<\frac{1}{n}$ and $R_2=(1-\alpha)(1-\frac{1}{n^2})$, and we have shown that $\Delta S_2 \le \frac{1-\alpha}{2}$.
Since $R_1<\frac{1}{n}$, agents in $G_1$ receive at most $\frac{1}{n}$ of resource~2 in step 2.
Thus, at least $R_2-\frac{1-\alpha}{2}-\frac{1}{n} \ge \frac{1-\alpha}{2}-\frac{2}{n}$ of resource~2 is not used under $\FTwo$.
Then 
\[
\AR_{\Util}(\FTwo) \ge 
\frac{1-\frac{2}{n}}{1-\frac{1-\alpha}{2}-\frac{2}{n}} \overset{n \to \infty}{\longrightarrow} \frac{1}{1-\frac{1-\alpha}{2}}.
\]

Under $\FTwos$, we have shown that $\Delta S^*_2 \le \Delta S_2 \le \frac{1-\alpha}{2}$.
Then using the same argument for $\FTwo$, at least $R_2-\frac{1-\alpha}{2}-\frac{1}{n} \ge \frac{1-\alpha}{2}-\frac{2}{n}$ of resource~2 is not used under $\FTwos$ and we get the same lower bound $\frac{1}{1-\frac{1-\alpha}{2}}$ for $\FTwos$.
\end{proof}

\myparagraph{Example 1 (continued).}
We compare $\FTwo$ and $\FTwos$ for the instance in Example 1.
Step 1 is the same as before and we have $\frac{R_1}{R_2}=\frac{\frac{4}{15}}{\frac{7}{15}}=\frac{4}{7}$ and $\frac{R_1^*}{R_2^*}=\frac{\frac{4}{15}+\frac{1}{15}}{\frac{7}{15}+\frac{1}{15}}=\frac{5}{8}$.
In step 2, under $\FTwo$, we increase the allocation of agent $2$ by $(\frac{16}{81},\frac{16}{405})$, and that of agent $3$ by $(\frac{28}{405},\frac{28}{81})$ such that resource $r_1$ is used up.
Notice that $\Delta S_1=\frac{16}{81}$ and $\Delta S_2=\frac{28}{81}$ satisfy $\frac{\Delta S_1}{\Delta S_1}=\frac{R_1}{R_2}$.
Under $\FTwos$, we increase the allocation of agent $2$ by $(\frac{20}{99},\frac{4}{99})$, and that of agent $3$ by $(\frac{32}{495},\frac{32}{99})$.
Notice that $\Delta S_1=\frac{20}{99}$ and $\Delta S_2=\frac{32}{99}$ satisfy $\frac{\Delta S_1}{\Delta S_1}=\frac{R_1^*}{R_2^*}$.
These two corresponds to Figure \ref{fig:example_sup23} and \ref{fig:example_sup24}.
The SW under $\FTwo$ and $\FTwos$ is $\approx 1.54$ and $\approx 1.53$ respectively, which is larger than $1.36$ under DRF and $1.47$ under $\FOne$.
\medskip

Figure \ref{fig:compare-ratio} shows \fratio{s} of DRF, $\FOne$, $\FTwo$, and $\FTwos$ (when $n \to \infty$) as a function of $\alpha$.
Notice that all three new mechanisms have better \fratio than DRF for any $\alpha \in (0,\frac{1}{2})$.
Among new mechanisms, $\FOne$ has better \fratio than $\FTwo$ ($\FTwos$) when $\alpha$ is close to 0 while $\FTwo$ ($\FTwos$) has better \fratio than $\FOne$ when $\alpha$ is close to $0.5$.
    Note that we can combine these two mechanisms  to achieve a better \fratio, which will be further discussed in Section~\ref{sec:PoSP}.

\begin{figure}[t]
    \centering
    \includegraphics[width=0.7\textwidth]{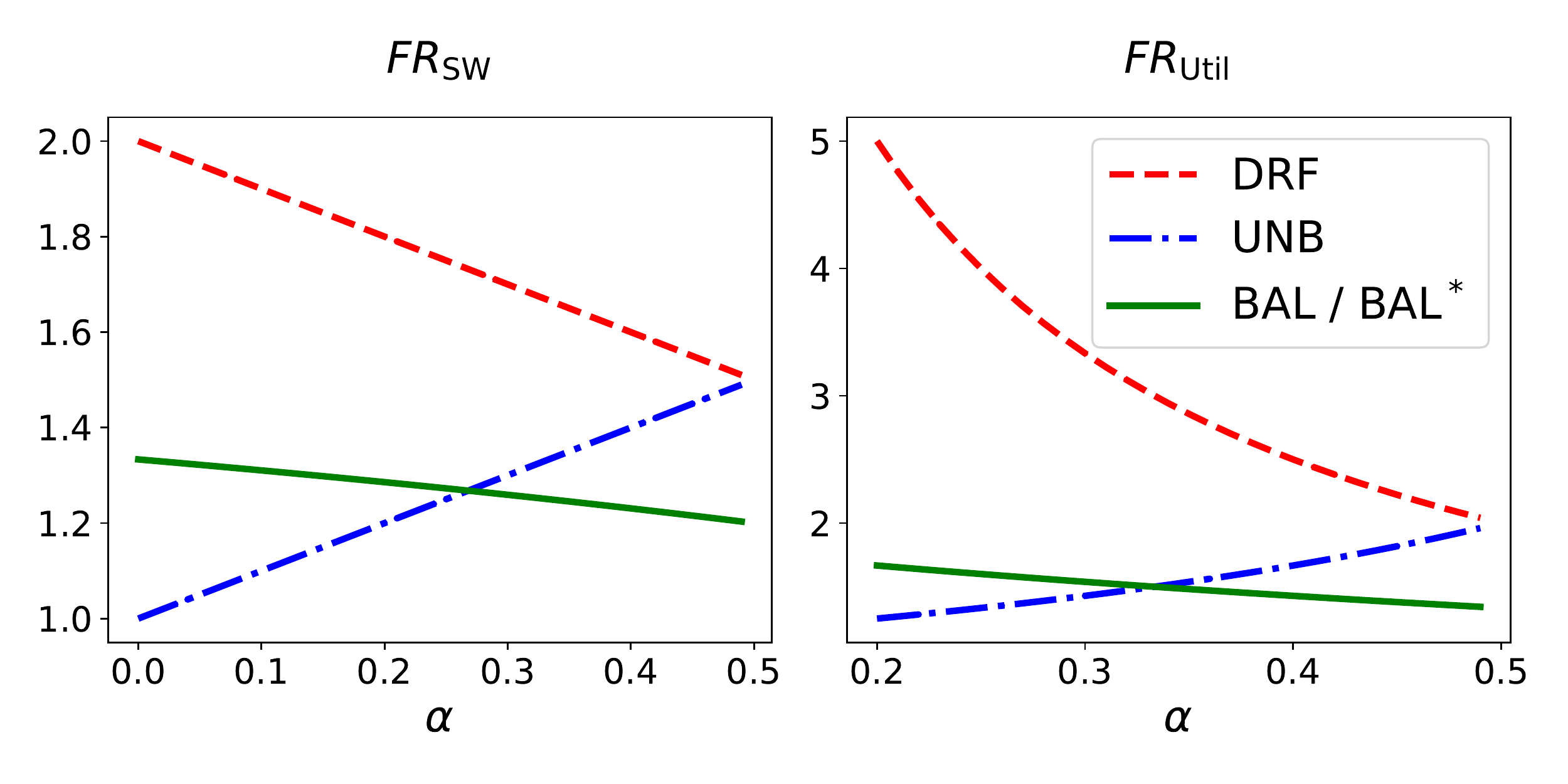}
    \caption{\Fratio of mechanisms as a function of $\alpha$. As $\AR_{\Util}(\DR) \to \infty$ when $\alpha \to 0$, for better visualization, we only show $\AR_{\Util}$ for $\alpha \in [0.2,0.5]$.}
    \label{fig:compare-ratio}
\end{figure}




\subsection{Experimental evaluation}
\label{sec:exp-m=2}

The above analysis of \fratio shows that our mechanism $\FOne$ and $\FTwos$ have better performance than DRF from the worst-case perspective.
In this section, we compare the performance of DRF, $\FOne$ and $\FTwos$ when $m=2$ using both synthetic instances and real-world instances based on Google cluster-usage traces~\cite{reiss2011google}.
Our results are shown in Figure~\ref{fig:compare-exp}, where we plot the ratio between the optimal allocation (satisfying SI and EF) and the allocation under compared mechanisms.
Our results match well with the above \fratio{s} and show that both $\FOne$ and $\FTwos$ achieve better social welfare and utilization than DRF.

\myparagraph{Random instances with different $\alpha$.}
First we compare mechanisms on random instances with fixed $n=100$ and different $\alpha \in \{0.05,0.10,\dots,0.50\}$.
For each $\alpha$, we average over 1000 instances to get the data point.
To control the value of $\alpha$, we choose $n(1-\alpha)$ agents and set $d_{i,1}=1$ for them, and for the remaining agents we set $d_{i,2}=1$.
The other entries of the demand vectors are sampled uniformly from $\{0.01,0.02,\dots,1.00\}$.

The result is shown in the first row of Figure \ref{fig:compare-exp}.
For SW, $\FTwos$ is very close to the optimal solution (the ratio is close to 1) and $\FTwos$ is always better than DRF for different values of $\alpha$.
$\FOne$ also outperforms DRF for most values of $\alpha$ except when $\alpha \in [0.45,0.5]$.
Comparing $\FOne$ and $\FTwos$, similarly to the crossing point of their theoretical \fratio{s} in Figure \ref{fig:compare-ratio}, their performance on random instances also cross when $\alpha \approx 0.25$ in Figure \ref{fig:compare-exp}, confirming that when $\alpha \to 0$, $\FOne$ is better than $\FTwos$, and when $\alpha \to 0.5$, $\FTwos$ is better than $\FOne$.
When $\alpha \ge 0.2$, the performance trend of three mechanisms matches well with the \fratio.
More precisely, when $\alpha$ increases, $\FTwos$ and DRF perform better while $\FOne$ performs worse.
The comparison of three mechanisms in utilization is almost the same as in SW.

\myparagraph{Instances generated from Google trace.}
Next we test mechanisms on instances that are generated according to the real demands of tasks from the Google traces.
The Google traces record the demands for CPU and memory of each submitted task.
We normalize these demands to get a pool of normalized demand vectors.
Then we generate instances by randomly sampling demand vectors from this pool.
We compare mechanisms on instances with different number $n \in \{10,20,\dots,100\}$ of agents.
For each $n$, we average over 1000 instances to get the data point.

The result is shown in the second row of Figure \ref{fig:compare-exp}.
For both SW and utilization, $\FOne$ and $\FTwos$ outperform DRF and the improvements are more than $10\%$.
The performance of $\FOne$ and $\FTwos$ are very close, because in the demand vector pool more agents (about $67\%$) have CPU as the dominant resource and hence the generated instances have $\alpha$ close to $0.33$.
Notice that the \fratio{s} of $\FOne$ and $\FTwos$ are indeed very close when $\alpha=0.33$ (see Figure \ref{fig:compare-ratio}).

\begin{figure}[t]
    \centering
    \includegraphics[width=0.7\textwidth]{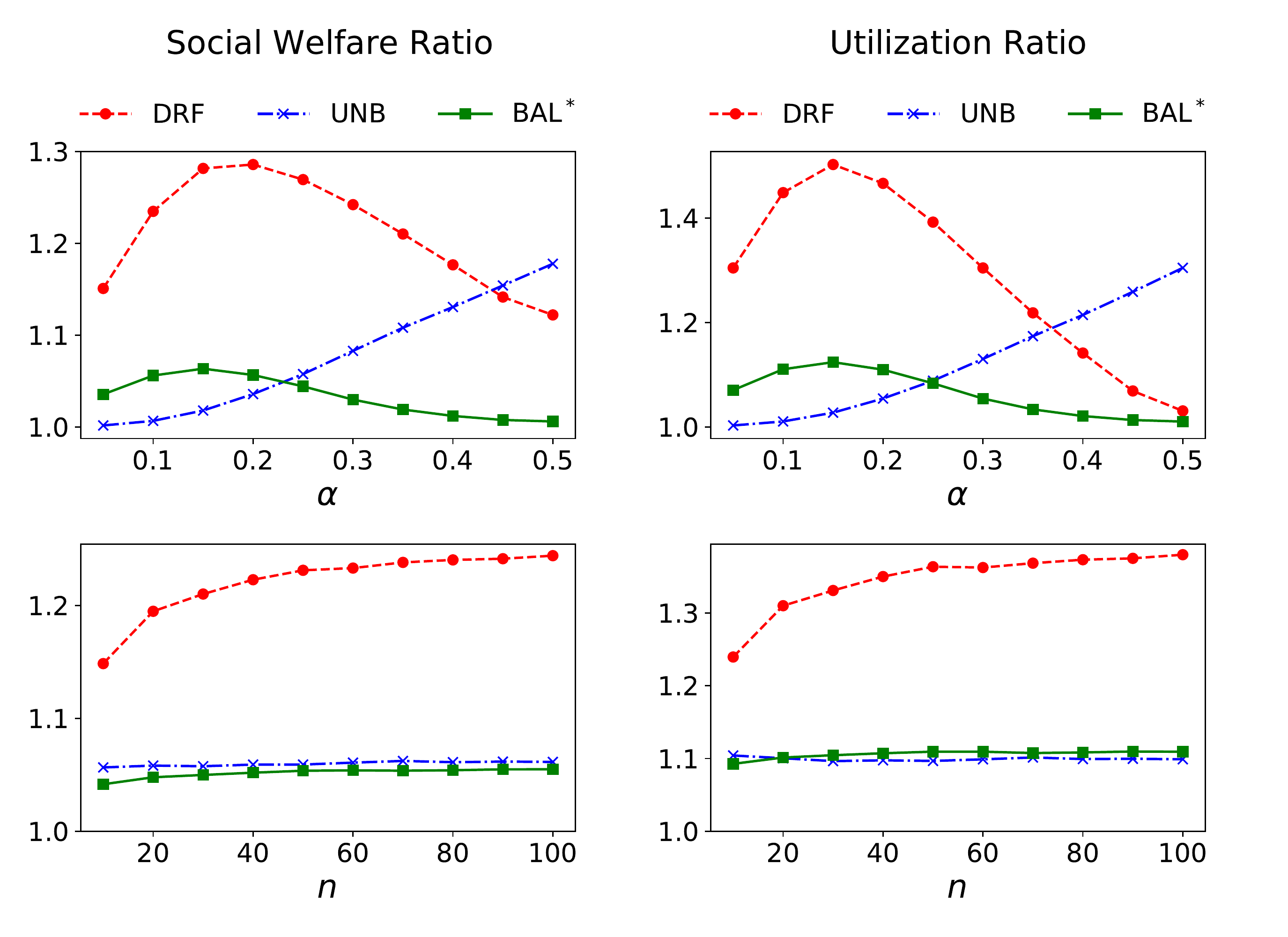}
    \caption{Performance ratio between the optimal allocation and allocations under DRF, $\FOne$, $\FTwos$ on synthetic instances with different $\alpha$ (1st row) and real-world instances with different $n$ (2nd row).
     }
    \label{fig:compare-exp}
\end{figure}

%

\section{Multiple Types of Resources}
We move to the general case with $m \ge 2$ types of resources.

\subsection{A family of mechanisms}
We start by presenting a large family of mechanisms that satisfy the four desired properties SI, EF, PO, and SP, which includes DRF as a special case.
This is in response to the question asked in \cite{Ghodsi2011} that \emph{``whether DRF is the only possible strategy-proof policy for multi-resource fairness, given other desirable properties such as Pareto efficiency''}.
Although many mechanisms based on DRF have been proposed for different settings and there are characterizations of mechanisms satisfying desirable properties under Leontief preferences \cite{Nicolo2004,Friedman2011,Li2013}, to the best of our knowledge, there is no work that directly answers this question.

We call a function $g$ that maps vectors $\mathbf{v} \in [0,1]^m$ to $\mathbb{R}$ $\emph{monotone}$, if it satisfies that for any two vectors $\mathbf{v}_1,\mathbf{v}_2$ with $\mathbf{v}_1>\mathbf{v}_2$, $g(\mathbf{v}_1)>g(\mathbf{v}_2)$.
Here $\mathbf{v}_1>\mathbf{v}_2$ means $\mathbf{v}_1$ is element-wise strictly larger than $\mathbf{v}_2$.
Denote $\mathcal{G}$ the set of all monotone functions.

Now we define a family of mechanisms $\mathcal{F}$ based on monotone functions.
For each monotone function $g \in \mathcal{G}$, we define a mechanism $f_g \in \mathcal{F}$ as follows.
The mechanism contains two steps that have the same flavor as $\FOne$.
In step 1, every agent receives $\frac{1}{n}$ dominant share.
In step 2, we increase the allocation for agents that have the minimum value of $g(\mathbf{A}_i)$ till some resource is used up.
The specific implementation of $F_g$ is shown in Algorithm \ref{alg:F-family}. 

For  mechanism $F_g\in\mathcal{F}$ which can run Steps \ref{Fgstep1} and \ref{Fgstep2} efficiently (i.e., in polynomial time), as $|P|$ is increasing in each round of step 2, so the number of rounds of step 2 is at most $n$, which means this mechanism $F_g$ can be implemented in polynomial time.

We show that all mechanisms in $\mathcal{F}$ satisfy the four desired properties.

\begin{algorithm}[t]
    \DontPrintSemicolon
    $\mathbf{C} \leftarrow (c_1, \ldots, c_m) = (1,\dots,1)$  \tcp*{remaining resources}
    \For{$1 \le i \le n$}
    {
      $\mathbf{A}_i \leftarrow \frac{1}{n} \mathbf{d}_i$ \tcp*{every agent receives $\frac{1}{n}$ dominant share initially}
      $\mathbf{C} \leftarrow \mathbf{C}-\mathbf{A}_i$
    }
    \While{$\forall i, c_i > 0$}
    {
      $P \leftarrow \arg\min_{i \in N}\{g(\mathbf{A}_i)\}$ \tcp*{the set of agents with smallest value of $g()$}
      \If{$P \neq N$}
      {
        $t \leftarrow \min_{i \in N \setminus P} g(\mathbf{A}_i)$ \tcp*{the second smallest value of $g()$}
      }
      \Else
      {
        $t \leftarrow g(1,\dots,1)$
      }
      find vector $\boldsymbol{\delta_0}$ such that $g((1+\delta_{0,i}) \cdot \mathbf{A}_{i})=t$ for all $i\in P$ \tcp*{increasing step when $t$ is reached} \label{Fgstep1}
      \ForEach {$j \in R$}
      {
          find vector $\boldsymbol{\delta_j}$ such that $g((1+\delta_{j,i}) \cdot \mathbf{A}_{i})=g((1+\delta_{j,i^{\prime}}) \cdot \mathbf{A}_{i^{\prime}})$ for all $i,i^{\prime}\in P$ and $\sum_{i\in P} \delta_{j,i}\cdot A_{i,j}=c_j$ \label{Fgstep2}
        \tcp*{increasing step when resource $j$ is used up}
      }
      $\boldsymbol{\delta^*} \leftarrow  \min\{\boldsymbol{\delta_0}, \min_{j \in R} \boldsymbol{\delta_j}\}$ \\
      \ForEach {$i \in P$}
        {
          $\mathbf{A}_i \leftarrow (1+\delta^*_i) \cdot \mathbf{A}_i$;
          \tcp*{increase allocations for agents in $P$}
          $\mathbf{C} \leftarrow \mathbf{C} - \delta^*_i \cdot \mathbf{A}_i$;
        }
    }
    \Return $\mathbf{A}$
    \caption{$F_g(\mathbf{d}_1,\mathbf{d}_2,\dots,\mathbf{d}_n)$}
    \label{alg:F-family}
    \end{algorithm}



\begin{restatable}{theorem}{FG}
\label{thm:F-family}
For any $m$, every mechanism $F_g \in \mathcal{F}$ satisfies SI, EF, PO, and SP.
\end{restatable}

\begin{proof}
SI and PO are clearly satisfied.
Next we show EF.
Suppose there exists a mechanism $F_g \in \mathcal{F}$ that is not EF.
Let $i$ and $i'$ be two agents such that $i$ envies $i'$ in an allocation $\mathbf{A}$ produced by $F_g$, i.e., $u_i(\mathbf{A}_{i'})>u_i(\mathbf{A}_{i})$.
Then we have $\mathbf{A}_{i'} > \mathbf{A}_i$ and hence $g(\mathbf{A}_{i'})>g(\mathbf{A}_i)$.
Let $t=\min_{j\in N} g(\mathbf{A}_j)$ and $P=\{j \in N \mid g(\mathbf{A}_j)=t\}$.
Since $g(\mathbf{A}_{i'})>g(\mathbf{A}_i)$, we have $i' \not \in P$ and then $\mathbf{A}_{i'}=\frac{1}{n} \mathbf{d}_{i'}$.
Let $k^*$ be the dominant resource of agent $i$, we have
\[
A_{i,k^*} \ge \frac{1}{n} \ge \frac{1}{n} d_{i',k^*} = A_{i',k^*},
\]
which contradicts with $\mathbf{A}_{i'} > \mathbf{A}_i$.
This finishes the proof for EF.

Finally we show SP.
Suppose there exists a mechanism $F_g \in \mathcal{F}$ that is not SP.
Let $i^*$ be the agent who can benefit by reporting a false demand vector $\mathbf{d}'_{i^*}$ instead of the true demand vector $\mathbf{d}_{i^*}$ in an instance $\mathbf{I}$.
Denote the truthful outcome by $\mathbf{A}$ and the manipulated outcome by $\mathbf{A}'$.
Let $t=\min_{i\in N} g(\mathbf{A}_i)$ and $P=\{i \in N \mid g(\mathbf{A}_i)=t\}$.
Let $t'$ and $P'$ be the corresponding notations for $\mathbf{A}'$.
Note that for any agent $i \in N \setminus \{i^*\}$,
since the allocation is non-wasteful, we have
$\mathbf{A}'_i=\lambda \mathbf{A}_i$ for some $\lambda \ge 0$ and hence
\[
u_i(\mathbf{A}'_i)>u_i(\mathbf{A}_i) \Leftrightarrow \mathbf{A}'_{i} > \mathbf{A}_{i} \Leftrightarrow g(\mathbf{A}'_i)>g(\mathbf{A}_i),
\] 
where $\mathbf{A}'_{i} > \mathbf{A}_{i}$ means that $A'_{i,j}>A_{i,j}$ for all $j \in R$. 
For agent $i^*$, the second part of the above formula only holds the direction from left to right, i.e., 
\[
u_{i^*}(\mathbf{A}'_{i^*})>u_{i^*}(\mathbf{A}_i) \Leftrightarrow \mathbf{A}'_{{i^*}} > \mathbf{A}_{{i^*}} \Rightarrow g(\mathbf{A}'_{i^*})>g(\mathbf{A}_{i^*}).
\]
Then from $u_{i^*}(\mathbf{A}'_{i^*})>u_{i^*}(\mathbf{A}_{i^*})$ we get $g(\mathbf{A}'_{i^*})>g(\mathbf{A}_{i^*})$ and $\mathbf{A}'_{i^*} > \mathbf{A}_{i^*}$.
Since $\mathbf{A}'_{i^*} > \mathbf{A}_{i^*} \ge \frac{1}{n} \mathbf{d}_{i^*}$, it must be the case that $i^* \in P'$ and $t \le g(\mathbf{A}_{i^*}) < g(\mathbf{A}'_{i^*}) =t'$.
Consequently, 
for any $i \in P$ we have $g(\mathbf{A}'_{i}) \ge t' \ge t = g(\mathbf{A}_{i})$, and for any $i \in N \setminus P$ we have $g(\mathbf{A}'_{i}) \ge g(\frac{1}{n}\mathbf{d}_i) = g(\mathbf{A}_{i})$.
Thus, we have $g(\mathbf{A}'_{i}) \ge g(\mathbf{A}_{i}) \Leftrightarrow \mathbf{A}'_{i} \ge \mathbf{A}_{i}$ for all $i \in N \setminus \{i^*\}$ and $\mathbf{A}'_{i^*} > \mathbf{A}_{i^*}$.
This contradicts with that $\mathbf{A}$ is PO.
This finishes the proof for SP.
\end{proof}

With the large family of mechanisms at hand, the next question is to check if there exists any mechanism from $\mathcal{F}$ that can achieve better efficiency than DRF. Unfortunately, as we will see in the next part, all mechanisms from $\mathcal{F}$ will have the same  approximation guarantee for general $m$. This means from a worst-case analysis point of view, no mechanism has a provable better SW or utilization than DRF. Thus a more fine-grained analysis is needed to find better mechanisms.
In the next section, we analyze a special mechanism from $\mathcal{F}$, which can be seen as a generalization of $\FOne$, by considering two parameters.

\subsection{Generalization of $\FOne$}
\label{sec:general-F1}

Similar to the case with $2$ resources, we first partition all agents into $m$ groups $G_i (i \in [m])$ according to their dominant resources and choose an arbitrary group (say $G_1$) as a special group.
Then, we let
$\alpha \coloneqq 1-\frac{|G_1|}{n}$
be the fraction of agents not in $G_1$, and let
$\beta \coloneqq \sum_{i \in N  \setminus G_1} \frac{d_{i,1}}{n\alpha}$
be the average demand of agents not in $G_1$ for resource $r_1$.

$\FOne$ can be generalized as follows.
In step 1, each agent gets $\frac{1}{n}$ dominant share. In step 2, we increase the allocation of agents who have the smallest fraction of resource $r_1$ in the same speed for resource $r_1$, till some resource is used up.
With slight abuse of notation, we still call this generalized mechanism $\FOne$.
Note that this mechanism is equivalent to the mechanism from the family $\mathcal{F}$ with monotone function $g(\mathbf{v})=\mathbf{v}_1$.
We prove the \fratio of $\FOne$ and DRF parameterized by $\alpha$ and $\beta$ in the following theorem.

\begin{figure}[t]
\centering
    \includegraphics[scale=0.4]{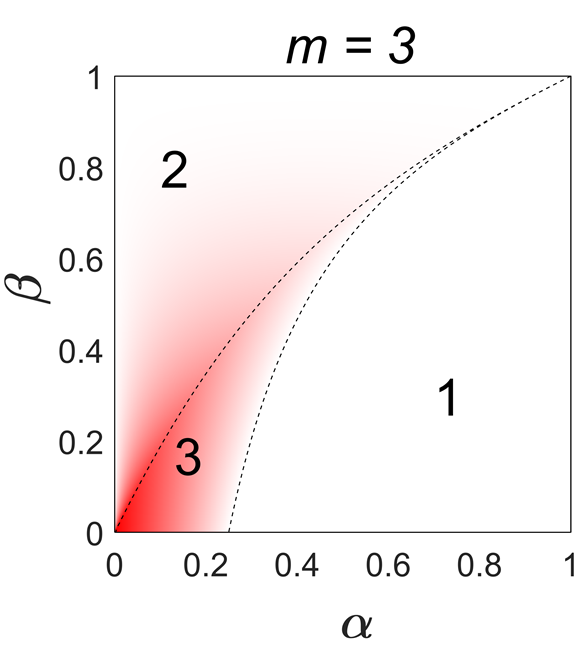}
   \quad
    \includegraphics[scale=0.4]{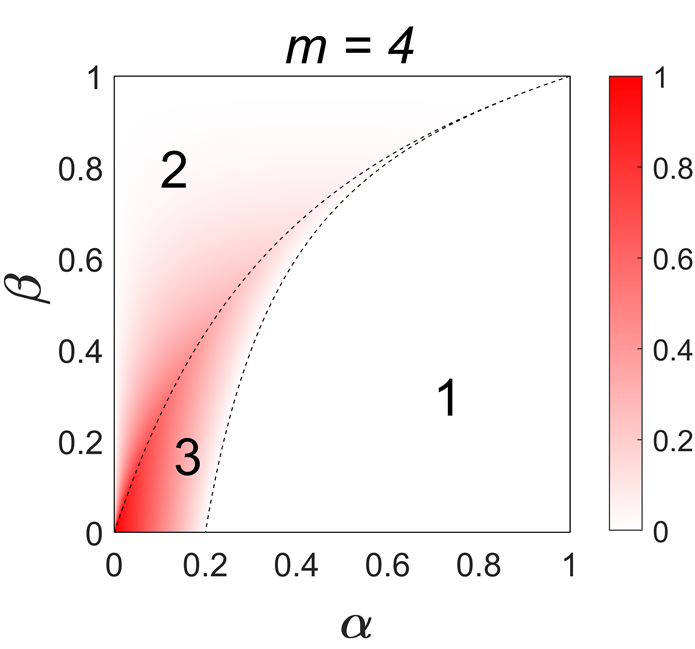}
  \caption{The difference of $\AR_{\SW}(\DR)$ over $\AR_{\SW}(\FOne)$.
  }
  \label{fig:mresbound}
\end{figure}

\begin{restatable}{theorem}{GFONE}
\label{thm:n-m-F1}
With $m \ge 3$ resources, mechanism $\FOne$ can be implemented in polynomial time, satisfies SI, EF, PO, and SP, has
$\AR_{\Util}(\FOne)=\AR_{\Util}(\DR)=\infty$, and
\[
\AR_{\SW}(\FOne)  = \max\left\{m-\alpha\beta-(1-\alpha),\frac{m-\alpha\beta}{1 + \frac{1-\beta}{\beta} \alpha}\right\},
\]
compared to
\[
\AR_{\SW}(\DR) = \max\left\{m-\alpha\beta-(1-\alpha), (m-\alpha\beta)(1-\alpha(1-\beta))\right\}.
\]
\end{restatable}

\begin{proof}
Since $\FOne$ is equivalent to the mechanism from the family $\mathcal{F}$ with monotone function $g(\mathbf{v})=\mathbf{v}_1$, according to Theorem \ref{thm:F-family}, we have that $\FOne$ satisfies SI, EF, PO, and SP and can be implemented in polynomial time. 
Then it suffices to show the \fratio{s}.
Let $\mathbf{A}$ be the allocation under $\FOne$.
We differentiate two cases according to whether there exists a resource other than resource~1 that is used up in $\mathbf{A}$.

We first consider the case when there is a resource other than resource~1 that is used up in $\mathbf{A}$. 
Assume this resource is resource~2.
Denote $x=\sum_{i \in G_1} A_{i,2}=\sum_{i \in G_1} \frac{d_{i,2}}{n}$. Note that $\sum_{i \in G_1} A_{i,1}=\sum_{i \in G_1} \frac{1}{n}=1-\alpha$ and $x \le \sum_{i \in G_1} A_{i,1}=1-\alpha$.
Then
\begin{align*}
\SW(A)& =\sum_{i \in G_1} A_{i,1}+\sum_{j \in \{2,\dots,m\}} \sum_{i \in G_j} A_{i,j} \\
& \ge 1-\alpha + \sum_{j \in \{2,\dots,m\}} \sum_{i \in G_j} A_{i,2} \\
& =1-\alpha+1-x.
\end{align*}
On the other hand, in any allocation satisfying SI, agents in $G_1$ receive at least $x$ of resource 2 and agents outside $G_1$ receive at least $\alpha\beta$ of resource 1, so the SW of the optimal allocation satisfying SI is upper bounded by $m-x-\alpha\beta$, and we have
\[
\begin{aligned}
\AR_{\SW}(\FOne) & \le \frac{m-x-\alpha\beta}{1-\alpha+1-x}\\
& \le \frac{m-(1-\alpha)-\alpha\beta}{1-\alpha+1-(1-\alpha)}\\
& =m-(1-\alpha)-\alpha\beta,
\end{aligned}
\]
where the second inequality holds since $x \le 1 - \alpha$ and $m-\alpha\beta \ge 2 -\alpha\beta \ge 2-\alpha$.

To show the corresponding lower bound~$\AR_{\SW}(\FOne)\ge m-(1-\alpha)-\alpha\beta$, consider the following instance. We set one parameter $\varepsilon$ that is close to $0$.
Group~$G_1$ consists of $n(1-\alpha)$ agents, where all agents have the demand vector $(1,1-\varepsilon,\varepsilon,\dots,\varepsilon)$ except one special agent $i_1^*$ whose demand vector is $(1,\varepsilon,\varepsilon,\dots,\varepsilon)$.
Group~$G_2$ consists of $n\alpha-(m-2)$ agents, where all agents have the demand vector $(\beta^{\prime},1,\varepsilon,\dots,\varepsilon)$ except one special agent $i_2^*$ who has the demand vector $(\varepsilon^2,1,\varepsilon,\dots,\varepsilon)$.
Each of the remaining $m-2$ agents has a different dominant resource for the remaining $m-2$ resources and they demand $\varepsilon$ for all non-dominant resources.
Here, we set $\varepsilon=\frac{1}{n^2}$ and $\beta^{\prime}$ such that the average demand of agents not in $G_1$ for resource~1 is $\beta$, i.e., $\beta^{\prime}(1-\frac{m-1}{n\alpha})+\frac{\varepsilon^2+(m-2)\varepsilon}{n\alpha}=\beta$. 
Without loss of generality, we assume $\varepsilon<\frac{\beta^{\prime}}{n}$, which can be reached by setting $n$ large enough.
In step $2$ of $\FOne$, we increase the allocation of the special agent $i_2^*$ in $G_2$ till resource $2$ is used up. 
Since there is less than $\varepsilon(1-\alpha)+\frac{m-2}{n}$ of resource $2$ after step $1$, the SW under $\FOne$ is upper bounded by
\[
1+\frac{1-\alpha}{n^2}+\frac{m-2}{n} \overset{n \to \infty}{\longrightarrow} 1.
\]
On the other hand, we can build an allocation $A^*$ satisfying SI and EF as follows.
For agents not in $G_1 \cup G_2$, each gets $1-\frac{1}{n}$ dominant share.
Each agent in $G_2$ and $G_1 \setminus \{i_1^*\}$ gets $\frac{1}{n}$ dominant share.
Let us consider the remaining resource, which will be given to the special agent $i_1^*$ in $G_1$.
The remaining amount of resource 1 is at least $\alpha-\alpha\beta-\varepsilon(m-2)$, since only agents not in $G_1 \cup G_2$ get more than $\frac{1}{n}$ dominant share and they get at most $\varepsilon(m-2)$ of resource 1.
The remaining amount of resource 2 is at least $\frac{1}{n} \ge \varepsilon$ as no agent receives more than $\frac{1}{n}$ of resource 2.
The remaining amount of resource $k$ for $k \in [3,m]$ is at least $\frac{1}{n}-\frac{m-2}{n^2}-(n-1-(m-2))\frac{\varepsilon}{n} \ge \frac{1}{n}-\frac{m-1}{n^2} \ge \varepsilon$.
Since the demand vector of agent $i_1^*$ is $(1,\varepsilon,\varepsilon,\dots,\varepsilon)$, the bottleneck resource is resource 1 and agent $i_1^*$ can get at least $\alpha-\alpha\beta-\varepsilon(m-2)$ dominant share from the remaining resource.
Then the SW of this allocation is at least
\[
\begin{aligned}
& \alpha-\alpha\beta-\varepsilon(m-2)+\frac{n-1-(m-2)}{n}+(m-2)\frac{n-1}{n} \\
=& m-(1-\alpha)-\alpha\beta -(\frac{m-2}{n^2}+\frac{2m-3}{n}),
\end{aligned}
\]
which approaches to $m-(1-\alpha)-\alpha\beta$ when $n \to \infty$.
So $\AR_{\SW}(\FOne) \ge m-(1-\alpha)-\alpha\beta$.

The above instance also shows that $\AR_{\Util}(\FOne)=\infty$.
Indeed, under $\FOne$ the usage of resource $k$ for any $k \in [3,m]$ is at most $\frac{1}{n}+n\varepsilon \le \frac{2}{n}$.
In $A^*$, resource 1 is used up; resource 2 is used at least $\frac{(n-1-(m-2))(1-\varepsilon)}{n} \ge 1-\varepsilon-\frac{m}{n}$; resource $k$ for any $k \in [3,m]$ is used at least $1-\frac{1}{n}$.
So 
\[
\AR_{\Util}(\FOne) \ge \frac{1-\varepsilon-\frac{m}{n}}{\frac{2}{n}} \overset{n \to \infty}{\longrightarrow} \infty.
\]

For the second case when only resource~1 is used up,
since $\FOne$ always increases the allocations of agents with the smallest fraction of resource~1, we have that
\[
\begin{aligned}
\frac{\sum_{j \in \{2,\dots,m\}} \sum_{i \in G_j} A_{i,j}}{\sum_{j \in \{2,\dots,m\}} \sum_{i \in G_j} A_{i,1}} &\ge \frac{\sum_{j \in \{2,\dots,m\}} \sum_{i \in G_j} d_{i,j}}{\sum_{j \in \{2,\dots,m\}} \sum_{i \in G_j} d_{i,1}} \\
&=\frac{n\alpha}{n\alpha\beta}\\
&=\frac{1}{\beta}.
\end{aligned}
\]
Since resource 1 is used up, we have $\sum_{j \in \{2,\dots,m\}} \sum_{i \in G_j} A_{i,1}=\alpha$, 
then 
\[\sum_{j \in \{2,\dots,m\}} \sum_{i \in G_j} A_{i,j} \ge \frac{\alpha}{\beta}.\]
We can bound SW as
\begin{align*}
\SW(A)& =\sum_{i \in G_1} A_{i,1}+\sum_{j \in \{2,\dots,m\}} \sum_{i \in G_j} A_{i,j} \\
& \ge 1-\alpha + \frac{\alpha}{\beta}.
\end{align*}
On the other hand, in any allocation satisfying SI, agents not in $G_1$ receive at least $\alpha\beta$ of resource 1, so the SW of the optimal allocation satisfying SI is upper bounded by $m-\alpha\beta$, and we have
\[
\AR_{\SW}(\FOne) \le \frac{m-\alpha\beta}{ 1-\alpha + \frac{\alpha}{\beta}}=\frac{m-\alpha\beta}{1 + \frac{1-\beta}{\beta} \alpha}.
\]

To show the corresponding lower bound~$\AR_{\SW}(\FOne) \ge \frac{m-\alpha\beta}{1 + \frac{1-\beta}{\beta} \alpha}$, consider the following instance.
In $G_1$, all $n\alpha$ agents have the same demand vector $(1,\varepsilon,\dots,\varepsilon)$, where $\varepsilon$ is close to~0.
For each resource $j \in [2,m]$, $G_j$ consists of $\frac{n\alpha}{m-1}$ agents, where one agent $i^*_j$ has a special demand vector $(\frac{\beta^{\prime}}{\sqrt{n}},\varepsilon^2,\dots,\varepsilon^2,1,\varepsilon^2,\dots,\varepsilon^2)$, and the remaining agents have the same demand vector $(\beta^{\prime},\varepsilon,\dots,\varepsilon,1,\varepsilon,\dots,\varepsilon)$. Here, we set $\varepsilon=\frac{1}{n}$ and $\beta^{\prime}$ such that $\frac{m-1}{n\alpha}(\frac{\beta^{\prime}}{\sqrt{n}}+\beta^{\prime}(\frac{n\alpha}{m-1}-1))=\beta$.
The idea is that in the SW-maximizing allocation $A^*$ satisfying SI and EF, for each $j \in [2,m]$, agent $i^*_j$ should get almost all remaining resource~$j$ after ensuring the SI condition
and get the first resource close to $\frac{\beta^{\prime}}{\sqrt{n}}$, which will be larger than that ($\frac{\beta^{\prime}}{n}$) for other agents in $G_j$ who get $\frac{1}{n}$ dominant share.
Notice that this does not violate EF since agent $i^*_j$ receives at most $\varepsilon^2$ of resource $k$ for $k \neq 1,j$ while other agents in $G_j$ receive at least $\frac{\varepsilon}{n}=\varepsilon^2$ of resource $k$.
Then, all remaining resource~1 is assigned to $G_1$, which leads to an allocation with $\SW$ close to $m-\alpha\beta$.
However, under $\FOne$, we will equalize the amount of resource 1 received by all agents in $G_j$ for all $j \in [2,m]$, and when resource 1 is used up, all these agents get at most $\frac{m-1}{\sqrt{n}\beta^{\prime}}+\frac{\alpha}{\beta}$ dominant share totally.
Since agents in $G_1$ have dominant share $1-\alpha$, we have
\[
\AR_{\SW}(\FOne) \ge \frac{m-\alpha\beta}{1-\alpha+\frac{m-1}{\sqrt{n}\beta^{\prime}}+\frac{\alpha}{\beta}} \overset{n \to \infty}{\longrightarrow} \frac{m-\alpha\beta}{1 + \frac{1-\beta}{\beta} \alpha}.
\]

Combining the two cases, we have $\AR_{\SW}(\FOne) = \max\{m-(1-\alpha)-\alpha\beta,\frac{m-\alpha\beta}{1 + \frac{1-\beta}{\beta} \alpha}\}$.
Similarly, we can get the \fratio for DRF.
The lower bounds can be proved using the same instances for $\FOne$.
For upper bounds, the only difference is in the second case when only resource~1 is used up, where we need to upper bound the SW under DRF by $\frac{1}{1-\alpha+\alpha\beta}=\frac{1}{1-\alpha(1-\beta)}$,
so $\AR_{\SW}(\DR) \le (m-\alpha\beta)(1-\alpha(1-\beta))$.
\end{proof}

In particular, one can show that for any $0 < \alpha, \beta < 1$, $1-\alpha(1-\beta) > \frac{1}{1 + \frac{1-\beta}{\beta} \alpha}$. This means $\AR_{\SW}(\FOne)$ is always weakly better than $\AR_{\SW}(\DR)$.
Figure \ref{fig:mresbound} visualizes the difference between $\AR_{\SW}(\FOne)$ and $\AR_{\SW}(\DR)$ with $m=3,4$, together with the ranges of $\alpha$ and $\beta$ for the following three scenarios about choices in the formula of $\AR_{\SW}(\FOne)$ and $\AR_{\SW}(\DR)$:
\begin{enumerate}
\item Both maximums achieve with the first term, where $\AR_{\SW}(\FOne)=\AR_{\SW}(\DR)$.
\item Both maximums achieve with the second term, where $\AR_{\SW}(\FOne)<\AR_{\SW}(\DR)$.
\item $\AR_{\SW}(\FOne)$ achieves the first term and $\AR_{\SW}(\DR)$ achieves the second, where $\AR_{\SW}(\FOne)<\AR_{\SW}(\DR)$ and the difference is larger than that in case~2.
\end{enumerate}

\subsection{Experimental evaluation}

\begin{figure*}
    \centering
    \includegraphics[width=14cm]{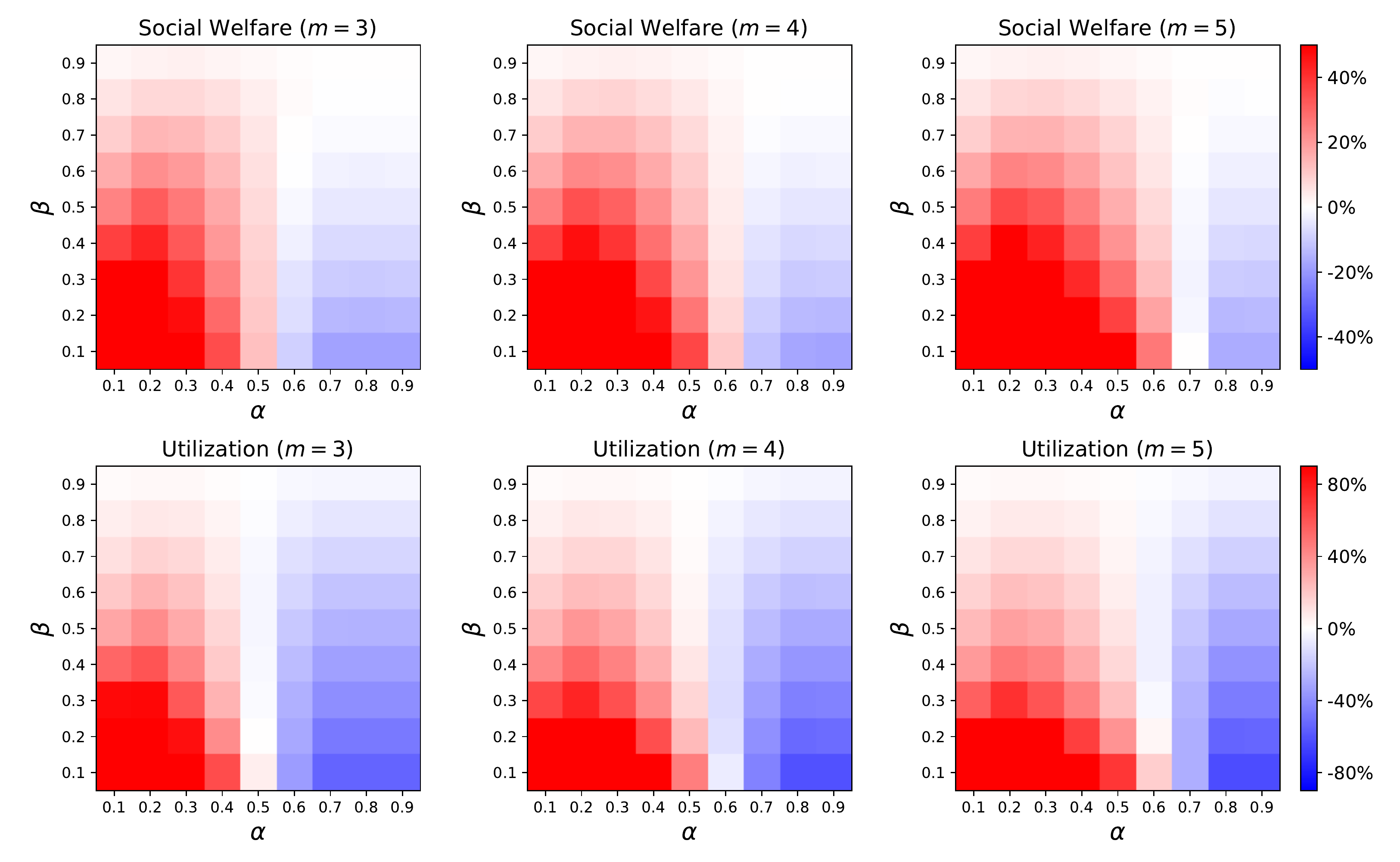}
    \caption{Improvement of $\FOne$ over DRF for $m=3,4,5$ and different values of $\alpha,\beta \in \{0.1,0.2,\dots,0.9\}$.}
    \label{fig:compare-F1G}
\end{figure*}

We compare $\FOne$ and DRF on random instances for $m=3,4,5$ and different values of $\alpha,\beta \in \{0.1,0.2,\dots,0.9\}$.
The instances are generated similarly as in Section \ref{sec:exp-m=2} with $n=100$ agents.
To control the value of $\alpha$, we choose $n(1-\alpha)$ agents and set their dominant resource as resource $1$ ($d_{i,1}=1$).
The dominant resource for the remaining $n\alpha$ agents are randomly chosen from the remaining resources.
To control the value of $\beta$, the demand for non-dominant resource of all agents are chosen from the distribution $\chi=(1-\beta)\mathcal{U}_1+\beta\mathcal{U}_2$, where $\mathcal{U}_1$ is the uniform distribution over $G \cap [0,\beta]$, $\mathcal{U}_2$ is the uniform distribution over $G \cap (\beta,1]$, where $G=\{0.01,0.02,\dots,1.00\}$.
For each configuration, we randomly generate 1000 problem instances and take the average of the results.

Figure \ref{fig:compare-F1G} shows the comparison of $\FOne$ and DRF.
For a wide range of $\alpha$ and $\beta$, $\FOne$ outperforms DRF in terms of both SW and utilization.
DRF has a better performance when $\alpha$ is close to 1, i.e., when the size of $G_1$ is small.
Notice that in practice we can choose the largest group as $G_1$ such that $\alpha < 1-\frac{1}{m}$.
In Figure \ref{fig:compare-F1G}, to give a fair visualization for positive and negative data, we do not differentiate the cases when $\FOne$ improves DRF by more than $40\%$ (resp. $80\%$) for social welfare (resp. utilization).
In fact, for social welfare, $\FOne$ increases that of DRF by $40\%$ to $100\%$ when $\alpha,\beta \le 0.3$, while in the worst case $\FOne$ is at most $20\%$ worse than DRF.
The difference in terms of utilization is even larger: $\FOne$ can increase that of DRF by more than $200\%$ in the best case. But in the worst case $\FOne$ is at most $70\%$ worse than DRF.

From these experiment results, whether $\FOne$ could outperform DRF in terms of social welfare and utilization depends on the structure of the problem instance. In particular, when $\alpha$ and $\beta$ are small, replacing DRF with $\FOne$ can significantly increase the social welfare and utilization of the solution.

\section{Price of Strategyproofness} \label{sec:PoSP}
At last, we investigate the efficiency loss caused by incentive constraints. In~\cite{Parkes2015}, it is shown that any mechanism that satisfies one of SI, EF and SP can only guarantee at most $1/m$ of the social welfare. However, the benchmark used in~\cite{Parkes2015} is the optimal social welfare among all allocations.
Recall that the \fratio defined in this paper is benchmarked against the best social welfare (resp. utilization) among all allocations that satisfy SI and EF. In other words, we are working entirely in the domain of fair allocations.
Moreover, note that any optimal allocation satisfying SI and EF must also satisfy PO since in our problem PO means at least one resource is used up.
Therefore, the lower bound of the \fratio characterizes the efficiency loss caused by strategyproofness.
Accordingly, we can define the price of strategyproofness as follows.

\begin{definition}
The \emph{Price of Strategyproofness (Price of SP)} for social welfare (resp. utilization) is defined as the best \fratio for social welfare (resp. utilization) among all mechanisms which satisfy SI, EF, PO and SP, i.e.,
\[
\textrm{Price of SP} = \min_{f \textrm{ is SI,EF,PO,SP}} \AR_{\SW}(f).
\]
\end{definition}





We now study the price of SP and start with the case with two resources.
Recall that $\AR_{\SW}(\FOne)$ is increasing with $\alpha$ while $\AR_{\SW}(\FTwos)$ is decreasing with $\alpha$ (see Figure \ref{fig:compare-ratio}).
By choosing the better one from $\FOne$ and $\FTwos$ for each value of $\alpha$, we get a new mechanism with a better \fratio than both $\FOne$ and $\FTwos$ and show that the price of SP is at most $3-\sqrt{3}$ for SW and at most  $\frac{3}{2}$ for utilization.

\begin{restatable}{theorem}{TWORSP}
\label{thm:n-2-SP}
With $2$ resources and $n\ge2$ agents, the price of SP is at most $3-\sqrt{3}+\frac{1}{2n} \overset{n \to \infty}{\longrightarrow}  3-\sqrt{3}$ for SW and at most $\frac{3}{2-\frac{1}{n}} \overset{n \to \infty}{\longrightarrow} \frac{3}{2}$ for utilization.
\end{restatable}

\begin{proof}
We first study SW.
We can combine $\FOne$ and $\FTwos$ to get a hybrid mechanism that satisfies SI, EF, PO, and SP, and has \fratio for SW at most $3-\sqrt{3}+\frac{1}{2n}$.
For any input instance, the hybrid mechanism applies one of $\FOne$ and $\FTwos$ according to the value of $\alpha$.
If $\alpha \le 2-\sqrt{3}+\frac{1}{2n}$, then $\FOne$ is applied; otherwise, $\FTwos$ is applied.
SI, EF, and PO are clearly satisfied by this hybrid mechanism.
We show that SP is also satisfied.
Note that for both $\FOne$ and $\FTwos$, if an agent reports the true demand vector, its utility is at least $\frac{1}{n}$, while if it reports a false demand vector such that it is partitioned into another group (consequently, $\alpha$ is changed), then its utility is at most $\frac{1}{n}$ as it can receive at most $\frac{1}{n}$ of its true dominant resource.
So agents can never benefit by switching a group and changing the value of $\alpha$.
Then, since both $\FOne$ and $\FTwos$ are SP, this hybrid mechanism is also SP.
From Theorem \ref{thm:n-2-F1} and \ref{thm:n-2-F2s} we have that
$\AR_{\SW}(\FOne)=1+\alpha$ and $\AR_{\SW}(\FTwos) \le \frac{2-\alpha}{1+\frac{1-\alpha}{2}-\frac{1}{2n}}$.
So the \fratio for SW of the hybrid mechanism is
\[
\max\{ \max_{0 \le \alpha \le 2-\sqrt{3}+\frac{1}{2n}} 1+\alpha, \max_{ 2-\sqrt{3}+\frac{1}{2n} < \alpha \le 1} \frac{2-\alpha}{1+\frac{1-\alpha}{2}}\} = 3-\sqrt{3}+\frac{1}{2n},
\]
which implies that the price of SP is at most $3-\sqrt{3}+\frac{1}{2n}$ for SW.

As for utilization, we can also combine $\FOne$ and $\FTwos$ to get a hybrid mechanism, but with a different switch point.
When $\alpha \le \frac{1}{3}+\frac{1}{3n}$, the hybrid mechanism applies $\FOne$; otherwise, it applies $\FTwos$.
Clearly all four properties are satisfied by this hybrid mechanism.
From Theorem \ref{thm:n-2-F1} and \ref{thm:n-2-F2s} we have that 
$\AR_{\Util}(\FOne)=\frac{1}{1-\alpha}$ and $\AR_{\Util}(\FTwos) \le \frac{1}{1-\frac{1-\alpha}{2}-\frac{1}{2n}}$.
So the \fratio for utilization of the hybrid mechanism is
\[
\max\{ \max_{0 \le \alpha \le \frac{1}{3}+\frac{1}{3n}} \frac{1}{1-\alpha}, \max_{ \frac{1}{3}+\frac{1}{3n} < \alpha \le 1} \frac{1}{1-\frac{1-\alpha}{2}}\} = \frac{3}{2-\frac{1}{2n}},
\]
which implies that the price of SP is at most $\frac{3}{2-\frac{1}{2n}}$ for utilization.
\end{proof}

We then show that for general $m$ except one special case, all mechanisms satisfying SI, EF, PO, and SP have the same \fratio.

\begin{restatable}{theorem}{LBOUND}
\label{thm:lower-bound}
For social welfare, the price of SP is $m$ when $m \geq 4$ and between 2 and 3 when $m=3$.
For utilization, the price of SP is $\infty$ for any $m \geq 3$.
\end{restatable}

One of the main results of \cite{Parkes2015} is that any mechanism that satisfies SP can only guarantee at most $1/m$ of the social welfare.
Theorem \ref{thm:lower-bound} strengthens this by showing that the result still holds for $m \ge 4$ even if we use \fratio as our benchmark.
For the case when $m=3$, we show that the price of SP is still $\infty$ for utilization, while for SW we can only get a lower bound of 2.
We leave the gap between 2 and 3 as an open question.

We prove Theorem \ref{thm:lower-bound} via the following Lemma \ref{lem:m>=4-SW} (for $m \ge 4$) and Lemma \ref{lem:m=3-SW} (for $m=3$).
Our proof follows a similar framework as \cite[Theorem 4.1]{Parkes2015}, but
requires a different construction to incorporate EF and SI.

\begin{lemma}
\label{lem:m>=4-SW}
With $m \ge 4$ resources, for any mechanism $f$ satisfying SI, EF, PO, and SP, we have
$\AR_{\SW}(f)=m$ and
$\AR_{\Util}(f)=\infty$.
\end{lemma}

\begin{proof}
We first construct an instance to show the result for SW, and then use the same instance to show the result for utilization.
For SW, we show that for any $m \ge 4$ and $\delta>0$, there exists a sufficiently large $n$ such that for any mechanism $f$, we have $\AR_{\SW}(f) \ge m-\delta$.

\paragraph{Instance construction.}
We construct an instance $\mathbf{I}$ with $n$ agents partitioned into $m$ groups.
The first group $G_1$ consists of $n_1$ agents who have the same demand vector $(1,\varepsilon_1,\dots,\varepsilon_1)$, where $\varepsilon_1=n^{-1}$.
For each $i \in [2,m]$, $G_i$ has $n_2$ agents, and the $j$-th agent in $G_i$ has a demand vector $(\frac{1}{2},0,\dots,0,1,\varepsilon_2\beta^j,\varepsilon_2\beta^{-j},0,\dots,0)$, where $1$ is on the $i$-th position, $\varepsilon_2=n^{-2n_2-3}$ and $\beta=n^2$.
Specifically, if $i=m-1$, the corresponding demand vector is $(\frac{1}{2},\varepsilon_2\beta^{-j},0,\dots,0,1,\varepsilon_2\beta^j)$, while the corresponding demand vector when $i=m$ is $(\frac{1}{2},\varepsilon_2\beta^j,\varepsilon_2\beta^{-j},0,\dots,0,1)$.
Note that for any $j \in [1,n_2]$, $\varepsilon_2\beta^j \le n^{-2n_2-3} \cdot n^{2n_2}=n^{-3}<\frac{\varepsilon_1}{n}$, so there is no envy from $G_1$ to $\bigcup_{i \in [2,m]} G_i$ in any allocation satisfying SI.
Moreover, for any $j_1<j_2$, we have $\varepsilon_2 \beta^{j_1}  < \frac{1}{n}\varepsilon_2 \beta^{j_2}$ and $\frac{1}{n} \varepsilon_2 \beta^{-j_1}  > \varepsilon_2 \beta^{j_2}$, so any allocation satisfying SI also satisfies EF within each group $G_i$ for $i \in [2,m]$.
Note that $n=n_1+(m-1)n_2$.

In the following proof we will change the demand vectors of some agents in the instance $\mathbf{I}$.
We make the restriction that when we change the demand vector of an agent from group $G_i (i \in [2,m])$, we can only change it by multiplying $d_{i,j}$ for all $j \neq i$ with the factor $\frac{n_2}{n}$.
For example,
\[
(\frac{1}{2},0,\dots,0,1,\varepsilon_2\beta^j,\varepsilon_2\beta^{-j},0,\dots,0)
\]
will be changed to
\[
(\frac{n_2}{n}\frac{1}{2},0,\dots,0,1,\frac{n_2}{n}\varepsilon_2\beta^j,\frac{n_2}{n}\varepsilon_2\beta^{-j},0,\dots,0).
\]
We call a $(m-1)$-tuple $(i_2,i_3,\dots,i_m)$ \emph{diverse} if $1 \le i_k \le n_2$ for every $k \in [2,m]$.
Let $T_{n_2,m}$ be the set of all diverse $(m-1)$-tuples.
Note that $|T_{n_2,m}|={n_2}^{m-1}$.

Next we show that (a) for any diverse tuple $(i_2,i_3,\dots,i_m) \in T_{n_2,m}$, if we change the $i_k$-th agent in $G_k$ for each $k \in [2,m]$, then there always exists an allocation that satisfies SI and EF, and has SW close to $m$; 
(b) for any mechanism satisfying SP and SI, we can always find a diverse $(m-1)$-tuple such that after the change, the SW will be close to 1.
Combining these two points, we get the claimed \fratio $m$ for SW.

\paragraph{Lower bound for the optimal fair allocation.}
Fix a $(i_2,i_3,\dots,i_m) \in T_{n_2,m}$ and change the demands of the corresponding agents.
Let $G^*$ be the set of agents that have been changed.
We build an allocation $\mathbf{A}^*$ as follows.
Every agent outside $G^*$ gets $\frac{1}{n}$ dominant share.
Every agent in $G^*$ gets the same dominant share $x$ such that all resources except resource~1 are exactly used up.
Notice that for resource~1, all agents in $G_1$ receive $\frac{n_1}{n}$; all agents in $\bigcup_{k \in [2,m]} G_k \setminus G^*$ receive $\frac{(m-1)(n_2-1)}{2n}$; all agents in $G^*$ receive $\frac{(m-1)n_2x}{2n}<\frac{(m-1)n_2}{2n}$ as $x<1$.
Sum them up and we get
\[
\frac{n_1}{n}+\frac{(m-1)(n_2-1)}{2n}+\frac{(m-1)n_2}{2n} < \frac{n_1+(m-1)n_2}{n}=1.
\]
So $\mathbf{A}^*$ is feasible and resource~1 is the only resource that is not used up.

We then compute $\SW(\mathbf{A}^*)$.
The main task is to compute $x$ since $\text{SW}(\mathbf{A}^*)=1-\frac{m-1}{n}+(m-1)x$.
For each resource $k \in [2,m]$, $G_1$ receives $\frac{n_1}{n}\varepsilon_1$; $G_{k-2}$ ($G_{m-1}$ when $k=2$ and $G_{m}$ when $k=3$) receives $y^-\coloneqq\frac{\varepsilon_2}{n}(\frac{\beta^{-1}(1-\beta^{-n_2})}{1-\beta^{-1}}-\beta^{-i_k}+x\cdot n_2\beta^{-i_k})$; $G_{k-1}$ ($G_m$ when $k=2$) receives $y^+\coloneqq\frac{\varepsilon_2}{n}(\frac{\beta(\beta^{n_2}-1)}{\beta-1}-\beta^{i_k}+x\cdot n_2\beta^{i_k})$; $G_k$ receives $\frac{n_2-1}{n}+x$.
Since except resource a all resources are used up, we have
\[
x=1-\frac{n_1}{n}\varepsilon_1-y^--y^+-\frac{n_2-1}{n}.
\]
Since $\varepsilon_1=\frac{1}{n}$, $\beta=n^2$ and $\varepsilon_2=n^{-2n_2-3}$, then $y^+ \le \frac{\varepsilon_2 (\beta^{n_2+1}+n_2\beta^{n_2})}{n} < \frac{1}{n}$ and $y^- < y^+ < \frac{1}{n}$.
Then
\[
x > 1-\frac{1}{n}-\frac{1}{n}-\frac{1}{n}-\frac{n_2-1}{n} = 1-\frac{n_2+2}{n},
\]
and
\[
\text{SW}(\mathbf{A}^*) =1-\frac{m-1}{n}+(m-1)x \ge m-\frac{(n_2+3)(m-1)}{n}.
\]

Next we show that $\mathbf{A}^*$ satisfies SI and EF.
SI is clearly satisfied.
So we just need to show EF.
Recall that for any $j \in [1,n_2]$, $\varepsilon_2\beta^j \le n^{-2n_2-3} \cdot n^{2n_2}=n^{-3}<\frac{\varepsilon_1}{n}$, so there is no envy from $G_1$ to $\bigcup_{k \in [2,m]} G_k$.
For each $k \in [2,m]$, each agent in $G_k$ receives at least $\frac{1}{n}$ of resource $k$ while every agent outside $G_k$ receives less than $\frac{1}{n}$ of resource $k$ (since $\max\{y^+,y^-,\varepsilon_1\} \le \frac{1}{n}$), so there is no envy from $G_k$ to agents outside $G_k$.
Then there is no envy between different groups and we just need to consider envy within each group.
Since each agent in $G^*$ receives $x>\frac{1}{n}$ of its dominant resource while every agent outside $G^*$ receives $\frac{1}{n}$ of its dominant resource, within each group $G_k$ for $k \in [2,m]$ we just need to check whether the $i_k$-th agent in $G_k$ is envied by some other $j$-th agent in $G_k$.
For any $k \in [2,m]$,
if $j>i_k$, then $j$ receives $\frac{1}{n}\varepsilon_2\beta^j$ of resource $k+1$(resource $2$ when $k=m$) and $i_k$ receives $x\frac{n_2}{n}\varepsilon_2\beta^{i_k}$ of resource $k+1$(resource $2$ when $k=m$).
Since $\beta=n^2>xn_2$, we have
\[
\frac{1}{n}\varepsilon_2\beta^j \ge \frac{\beta}{n} \varepsilon_2\beta^{i_k} > x\frac{n_2}{n}\varepsilon_2\beta^{i_k},
\]
and then $j$ receives more of resource $k+1$ than $i_k$ and hence there is no envy from $j$ to $i_k$.
Similarly, if $j < i_k$, then $j$ receives more of resource $k+2$ (resource $2$ when $k=m-1$ and resource $3$ when $k=m$) than $i_k$ and hence there is no envy from $j$ to $i_k$.

\paragraph{Upper bound for any SP and SI mechanism.}
We show that for any mechanism satisfying SP and SI, we can always find a diverse $(m-1)$-tuple such that after changing the demand vectors of the corresponding agents all of them will get less than $\frac{m^2}{n}$ of resource~1.
Consequently, the SW will be close to 1.
For any $k \in [2,m]$, let $L^k_{n_2,m} \subseteq T_{n_2,m}$ be the set of tuples such that after changing the demand vectors of the corresponding agents, the $i_k$-th agent from $G_k$ gets at least $\frac{m^2}{n}$ of resource~1.
Let $L_{n_2,m}=\bigcup_{k \in [2,m]} L^k_{n_2,m}$.
Our goal is to show that $|L_{n_2,m}| < |T_{n_2,m}|$.

Because of SI, agents in $G_1$ receive at least $\frac{n_1}{n}$ of resource~1, then all remaining agents receive at most $\frac{(m-1)n_2}{n}$ of resource~1.
Fix any $k \in [2,m]$.
For any tuple $t=(i_2,i_3,\dots,i_m) \in L^k_{n_2,m}$, let $t_{-k}=(i_2,\dots,i_{k-1},i_{k+1},\dots,i_m)$.
If we change the corresponding agents for $t_{-k}$, then the number of agents from $G_k$ that can receive at least $\frac{m^2}{n}$ of resource~1 is at most
\[
\frac{\frac{(m-1)n_2}{n}}{\frac{m^2}{n}} \le \frac{n_2}{m}.
\]
If the $i_k$-th agent receives less than $\frac{m^2}{n}$ of resource~1, then after changing the $i_k$-th agent, she still receives less than $\frac{m^2}{n}$ of resource~1, since otherwise she will benefit by changing its demand vector, which contradicts with SP.
So there are at most $\frac{n_2}{m}$ choices for $i_k$ for a fixed $t_{-k}$.
Then
\[
|L^k_{n_2,m}| \le \frac{n_2}{m} \cdot {n_2}^{m-2}=\frac{{n_2}^{m-1}}{m}, \text{ and }
\]
\[
  \begin{aligned}
|L_{n_2,m}| \le \sum_{k \in [2,m]} |L^k_{n_2,m}| &\le \frac{(m-1){n_2}^{m-1}}{m} \\
&< {n_2}^{m-1} \\
&= |T_{n_2,m}|.
  \end{aligned}
\]

Therefore, there exists at least one tuple in $T_{n_2,m}$ such that after changing the demand vectors of the corresponding agents, all of them will get less than $\frac{m^2}{n}$ of resource~1.
Then the SW of the allocation after the change is upper bounded with adding an assumption that $n\geq n^2_2$ by
\[
  \begin{aligned}
& \frac{n_1}{n}+\frac{2(m-1)n_2}{n}+(m-1) \cdot \frac{2m^2}{n_2} \\
&\le 1+\frac{(m-1)n_2}{n}+(m-1) \cdot \frac{2m^2}{n_2}\\
& \leq 1+\frac{3m^3}{n_2}.
  \end{aligned}
\]

Finally, since there always exists an allocation $\mathbf{A}^*$ with $\text{SW}(\mathbf{A}^*) \ge m-\frac{(n_2+3)(m-1)}{n}$, the \fratio for SW is lower bounded by
\[
\frac{m-\frac{(n_2+3)(m-1)}{n}}{1+\frac{3m^3}{n_2}}.
\]
For any $\delta$ and $m \ge 4$, we can choose $n_2 \ge \frac{4m^4}{\delta}$ and $n \ge n_2^2$ such that
\[
\frac{m-\frac{(n_2+3)(m-1)}{n}}{1+\frac{3m^3}{n_2}} \ge \frac{m-\frac{\delta}{4}}{1+\frac{3\delta}{4m}} \ge m-\delta.
\]
This finished the proof for SW.

\paragraph{Proof for utilization.}
We use the same instance as above.
Recall that under $\mathbf{A}^*$ all resources except resource~1 are used up.
Since $G_1$ consists of $n_1$ agents and each agent receives at least $\frac{1}{n}$ of resource~1, the utilization of resource~1 is at least $\frac{n_1}{n}$.
On the other hand, we have shown that for any mechanism satisfying SP and SI, we can find a diverse $(m-1)$-tuple such that after changing the demand vectors of the corresponding agents all of them will get less than $\frac{m^2}{n}$ of resource~1.
Let's consider the utilization of resource~2 in this case.
Agents in $G_1$, $G_2$, $G_{n-1}$, $G_n$ receive non-zero amount of resource~2.
First, the amount of resource~2 received by $G_1$ is at most $\frac{1}{n}$.
Next, for all agents in $G_2 \cup G_{n-1} \cup G_n$ except those in the $(m-1)$-tuple, they can use at most $1-\frac{n_1}{n}$ of resource~1, so they can use at most $2(1-\frac{n_1}{n})$ of resource~2 since in their demand vectors the demand for resource~1 is $\frac{1}{2}$.
Finally, for the three agents in $G_2 \cup G_{n-1} \cup G_n$ and the $(m-1)$-tuple, we have shown that each agent receives less than $\frac{m^2}{n}$ of resource~1 and hence less than $\frac{2m^2}{n_2}$ of resource~2, so in total they receive at most $\frac{6m^2}{n_2}$ of resource~2.
To sum up, the utilization of resource~2 is at most
$\frac{1}{n}+2(1-\frac{n_1}{n})+\frac{6m^2}{n_2}$.
For any large number $\gamma$, we can choose $n_2 \ge 36m^2\gamma$ and $n_1 \ge 6(m-1)n_2\gamma$ such that $\frac{1}{n_1} \le \frac{1}{3\gamma}$, $2(\frac{n}{n_1}-1) \le \frac{1}{3\gamma}$, and $\frac{6m^2n}{n_1n_2} \le \frac{12m^2}{n_2} \le \frac{1}{3\gamma}$, and then $\AR_{\Util}(f)$ is at least
\[
  \begin{aligned}
\frac{\frac{n_1}{n}}{\frac{1}{n}+2(1-\frac{n_1}{n})+\frac{6m^2}{n_2}}&=
\frac{1}{\frac{1}{n_1}+2(\frac{n}{n_1}-1)+\frac{6m^2n}{n_1n_2}} \\
&\ge
\frac{1}{\frac{1}{3\gamma}+\frac{1}{3\gamma}+\frac{1}{3\gamma}}\\
& \ge
\gamma.
  \end{aligned}
\]
This finished the proof for utilization.
\end{proof}

\paragraph{Remark.} We include $0$ in the demand vector for a better and clearer description of the proof. Precisely, we can replace $0$ with a small enough $\varepsilon$ to get the same result.

The remaining case is when $m=3$.
Notice that for the constructed instance in the proof of Lemma \ref{lem:m>=4-SW}, we need two resources to avoid envy between agents in the same group, and thus the proof only works for $m \ge 4$.
For $m=3$, we can use one resource to ``partly'' avoid envy between agents in the same group and use some further techniques to guarantee EF.
With this change, we can show that the price of SP for utilization is also $\infty$ when $m=3$.
However, we can only get the lower bound 2 for price of SP for SW when $m=3$.
We leave the ga between 2 and 3 as an open question.

\begin{lemma}
\label{lem:m=3-SW}
With $m=3$ resources, for any mechanism $f$ satisfying SI, EF, PO and SP, we have
$\AR_{\SW}(f) \in [2,3]$ and
$\AR_{\Util}(f)=\infty$.
\end{lemma}

\begin{proof}
We first construct an instance to show the result for SW, and then use the same instance to show the result for utilization.
For SW, the upper bound is trivial since any mechanism $f$ satisfying SI can guarantee at least SW of 1, and hence $\AR_{\SW}(f) \le 3$.
For the lower bound, we show that for any  $\delta\in(0,1)$, there exists sufficiently large $n$ such that for any mechanism $f$, we have $\AR_{\SW}(f) \ge 2-\delta$.

\paragraph{Instance construction.} 
We first construct an instance $\mathbf{I}$ with $n$ agents partitioned into $2$ groups.
The first group $G_1$ consists of $n_1$ agents who have the same demand vector $(1,\varepsilon_1,1)$, where $\varepsilon_1=n^{-1}$.
The second group $G_2$ consists of $n_2$ agents, and the $j$-th agent in $G_2$ has a demand vector $(\frac{1}{n_2},1,\varepsilon_2\beta^j)$, where $\varepsilon_2=n^{-2n_2-3}$ and $\beta=n^2$.
Note that for any $j \in [1,n_2]$, $\varepsilon_2\beta^j \le n^{-2n_2-3} \cdot n^{2n_2}=n^{-3}<\frac{1}{n}$, so there is no envy from $G_1$ to $G_2$ in any allocation satisfying SI.
Note that $n=n_1+n_2$ and keep $n>n_2^2$.

In the following proof we will change the demand vectors of agents in $G_2$.
We make the restriction that when we change the demand vector of the $j$-th agent in $G_2$, we can only change it by multiplying $d_{j,1}$ and $d_{j,3}$ with the factor $\gamma_j=\frac{n_2(n_2-1)}{jn}<1$.
For example,
\[
(\frac{1}{n_2},1,\varepsilon_2\beta^j)
\]
will be changed to
\[
(\gamma_j\frac{1}{n_2},1,\gamma_j\varepsilon_2\beta^j).
\]

Similar to the case for $m \ge 4$ in Lemma \ref{lem:m>=4-SW}, in the following we show that (a) for any $j \in [1,n_2]$, if we change the $j$-th agent in $G_2$, then there always exists an allocation that satisfies SI and EF, and has SW close to $2$;
(b) for any mechanism satisfying SP and SI, we can always find an agent from $G_2$ such that after changing its demand vector, the SW of the allocation under this mechanism will be close to 1.
Combining these two points, we get the claimed lower bound $2$ for SW.

\paragraph{Lower bound for the optimal fair allocation.}
Fix a $j \in [1,n_2]$ and change the demand vector of the $j$-th agent in $G_2$.
Let $G_2^j \subseteq G_2$ be the set of the first $j$ agents in $G_2$ and $G_2^{-j}=G_2 \setminus G_2^{j}$.
We build an allocation $\mathbf{A}^*$ as follows.
Every agent outside $G_2^j$ gets $\frac{1}{n}$ dominant share.
Every agent in $G_2^j$ gets the same fraction $x$ of resource~1 such that resource~2 is exactly used up.
We will show that in $\mathbf{A}^*$ resource~1 and resource 3 are not used up, and hence, $\mathbf{A}^*$ is feasible.
Let us compute the value of $x$ first.
For resource~2, all agents in $G_1$ receive $\frac{n_1}{n}\varepsilon_1$;
all agents in $G_2^{-j}$ receive $\frac{n_2-j}{n}$;
all agents in $G_2^j$ receive $n_2(j-1)x+\frac{n_2x}{\gamma_j}$.
Thus,
\[
\frac{n_1}{n}\varepsilon_1+\frac{n_2-j}{n}+n_2(j-1)x+\frac{n_2x}{\gamma_j}=1,
\]
and then
\[
x=\frac{1-\frac{n_1}{n}\varepsilon_1-\frac{n_2-j}{n}}{\frac{n_2}{\gamma_j}+n_2(j-1)}
=\frac{\gamma_j}{n_2}\frac{n-n_1\varepsilon_1-(n_2-j)}{n+(1-\frac{1}{j})n_2(n_2-1)}\leq \frac{\gamma_j}{n_2}.
\]
Now for resource~1, all agents in $G_1$ receive $\frac{n_1}{n}$;
all agents in $G_2^{-j}$ receive $\frac{n_2-j}{n_2n}$; all agents in $G_2^j$ receive $jx\leq j\frac{\gamma_j}{n_2}=\frac{n_2-1}{n}$.
Sum them up and we get
\[
\frac{n_1}{n}+\frac{n_2-j}{n_2n}+\frac{n_2-1}{n}<1.
\]
For resource 3, since $\varepsilon_2\beta^j \le n^{-3}<\frac{1}{n_2}$, we have that
every agent receives less resource 3 than resource~1 and resource~2, so resource 3 is not used up.
Therefore, $\mathbf{A}^*$ is feasible and resource~2 is the only resource that is used up.
The SW of $\mathbf{A}^*$ is
\[
\text{SW}(\mathbf{A}^*)=\frac{n_1}{n}+1-\frac{n_1}{n}\varepsilon_1=1-\frac{n_2}{n}+1-\frac{n_1}{n^2} \ge 2-\frac{n_2+1}{n}.
\]

Next we show that $\mathbf{A}^*$ satisfies SI and EF.
SI is clearly satisfied.
So we just need to show EF.
There is no envy within $G_1$ as all agents receive the same amount of resources, so does $G_2^{-j}$.
There is no envy within $G_2^j$ as all agents receive the same amount $x$ of resource~1.
Thus, envy could only happen between different groups.
There is no envy from $G_1$ since every agent receives more fraction of resource 3 than agents in $G_2$ ($\frac{1}{n}>n^{-3}\geq \varepsilon_2\beta^j \geq \gamma_j\varepsilon_2\beta^j$).
There is no envy from $G_2^j$ since they receive more fraction of resource~1 than others in other groups ($x \ge \frac{1}{n}$).
There is no envy from $G_2^{-j}$ to $G_1$ since agents in $G_2^{-j}$ receive more fraction of resource~2 than agents in $G_1$ ($\frac{1}{n}>\frac{\varepsilon_1}{n}$).
Finally, there is no envy from $G_2^{-j}$ to $G_2^j$ since every agent in $G_2^{-j}$ receives at least $\frac{1}{n}\varepsilon_2\beta^{j+1}=n\varepsilon_2\beta^{j}$ of resource 3 while every agent in $G_2^{j}$ receives at most $\varepsilon_2\beta^{j}$ of resource 3.

\paragraph{Upper bound for any SP and SI mechanism.}
We show that for any mechanism satisfying SP and SI, we can always find an agent from $G_2$ such that after changing its demand vector, the SW of the allocation under this mechanism will be close to 1.
Suppose that for each $j \in [1,n_2]$, the $j$-th agent in $G_2$ receives at least $\frac{1}{\ln n_2}\frac{1}{j}\frac{n_2}{n}$ of resource~1 in the original instance $\mathbf{I}$.
Then the sum of resource~1 received by agents in $G_2$ is at least
\[
\sum_{j \in [1,n_2]} \frac{1}{\ln n_2}\frac{1}{j}\frac{n_2}{n}>\frac{n_2}{n}.
\]
On the other hand, since the mechanism is SI, all agents in $G_1$ receive at least $\frac{n_1}{n}=1-\frac{n_2}{n}$ of resource~1, which leads to a contradiction.
Therefore, at least one agent from $G_2$ receives less than $\frac{1}{\ln n_2}\frac{1}{j}\frac{n_2}{n}$ of resource~1 in the original instance.
Let this agent be the $j^*$-th agent in $G_2$ and we change the demand vector of this agent.
Because of SP, after the change this agent still receives less than $\frac{1}{\ln n_2}\frac{1}{j}\frac{n_2}{n}$ of resource~1.
Thus, it receives at most $\frac{n_2}{n_2-1}\frac{1}{\ln n_2}<\frac{2}{\ln n_2}$ dominant share as we will choose $n_2>2$.
Now we upper bound the SW after the change.
Agents in $G_1$ receive $\frac{n_1}{n}$ dominant share.
All remaining agents except the $j^*$-th agent in $G_2$ receive at most $\frac{n_2}{n}$ of resource~1 and at most $\frac{{n_2}^2}{n}$ of dominant share.
Then the SW after the change is upper bounded by
\[
\frac{n_1}{n}+\frac{{n_2}^2}{n}+\frac{2}{\ln n_2}=1+\frac{n_2(n_2-1)}{n}+\frac{2}{\ln n_2}.
\]

Finally, since there always exists an allocation $\mathbf{A}^*$ with $\text{SW}(\mathbf{A}^*) \ge2-\frac{n_2+1}{n}$, the \fratio for SW is lower bounded by
\[
\frac{2-\frac{n_2+1}{n}}{1+\frac{n_2(n_2-1)}{n}+\frac{2}{\ln n_2}}.
\]
For any $\delta\in(0,1)$ we can choose $n_2 \ge e^{\frac{16}{\delta}}$ and $n \ge \frac{8}{\delta}{n_2}^2$ such that $\frac{2}{\ln n_2} \le \frac{\delta}{8}$ and $\frac{n_2+1}{n} \le \frac{n_2(n_2-1)}{n} \le \frac{\delta}{8}$, then
\[
\frac{2-\frac{n_2+1}{n}}{1+\frac{n_2(n_2-1)}{n}+\frac{2}{\ln n_2}} \ge \frac{2-\frac{\delta}{8}}{1+\frac{\delta}{4}} \ge 2-\delta.
\]
This finished the proof for SW.

\paragraph{Proof for utilization.}
We use the same instance as above.
In $\mathbf{A}^*$ resource~2 is used up, and the utilization rate for resource~1 and 3 is at least $\frac{n_1}{n}$ because of SI.
So the utilization under $\mathbf{A}^*$ is at least $\frac{n_1}{n}$.
On the other hand, we have shown that for any mechanism satisfying SP and SI, we can always find an agent from $G_2$ such that after changing its demand vector, this agent receives at most $\frac{2}{\ln n_2}$ dominant share.
Let's consider the utilization rate of resource~2 in this case.
Agents in $G_1$ receive at most $\varepsilon=\frac{1}{n}$ of resource~2.
Agents in $G_2$ except the chosen agent receive at most $1-\frac{n_1}{n}$ of resource~1 due to SI and at most $n_2(1-\frac{n_1}{n})$ of resource~2.
So the utilization rate of resource~2 is at most $\frac{2}{\ln n_2}+\frac{1}{n}+n_2(1-\frac{n_1}{n})$.
For any large number $\gamma$, we can choose $n_2 \ge e^{12\gamma}$ and $n_1 \ge 3\gamma n_2^2$ such that $\frac{1}{n_1} \le \frac{1}{3\gamma}$, $\frac{2n}{n_1 \ln n_2} \le \frac{4}{\ln n_2} \le \frac{1}{3\gamma}$,  and $\frac{n_2^2}{n_1} \le \frac{1}{3\gamma}$, and then $\AR_{\Util}(f)$ is at least
\[
  \begin{aligned}
\frac{\frac{n_1}{n}}{\frac{2}{\ln n_2}+\frac{1}{n}+n_2(1-\frac{n_1}{n})}&=
\frac{1}{\frac{2n}{n_1 \ln n_2}+\frac{1}{n_1}+\frac{n_2^2}{n_1}} \\
&\ge
\frac{1}{\frac{1}{3\gamma}+\frac{1}{3\gamma}+\frac{1}{3\gamma}} \\
&\ge
\gamma.
  \end{aligned}
\]
This finished the proof for utilization.
\end{proof}

\section{Conclusion}
In this paper, we investigate the multi-type resource allocation problem. Generalizing the classic DRF mechanism, we propose several new mechanisms in the two-resource setting and in the general $m$-resource setting. The new mechanisms satisfy the same set of desirable properties as DRF but with better efficiency guarantees.
For future works, we hope to extend these mechanisms to handle more realistic assumptions, such that when agents have limited demands or indivisible task.
Another extension is to model and study the multi-resource allocation problem in a dynamic setting.

%
%
%
\bibliographystyle{splncs04}
\bibliography{bib}
%


\end{document}